\documentclass[12pt]{article}
\usepackage[latin1]{inputenc}
\usepackage{amsmath}
\usepackage{amsfonts}
\usepackage{amssymb}
\usepackage[hidelinks]{hyperref}
\usepackage{bbm}
\usepackage{bm}
\usepackage{cite}
\usepackage{enumerate}
\usepackage{graphicx,epsf,subfigure,pstricks,psfrag,pst-node,amsthm,amssymb,amsmath}
\usepackage[authoryear]{natbib}
\usepackage{multirow}
\usepackage{tabularx}
\usepackage{enumerate}% http://ctan.org/pkg/enumerate
\usepackage{hyperref}
\usepackage{bm}
\usepackage{extarrows, framed}
\usepackage{accents}
\usepackage{enumitem}
\usepackage{pbox}

\usepackage{mathrsfs}
\usepackage{lastpage,xr,refcount,etoolbox}
\usepackage{natbib} % jasa template
\usepackage{booktabs}
\usepackage{color}

\usepackage{graphicx,psfrag,epsf} % jasa template

\newtheorem{theorem}{Theorem}
\newtheorem{lemma}{Lemma}

\newtheorem{proposition}{Proposition}

\newtheorem{assump}{Assumption}

\newcommand{\A}{{\bm  A}} 
\newcommand{\D}{{\bm  D}} 
\newcommand{\V}{{\bm V}} 
\newcommand{\X}{{\bm  X}} 
\newcommand{\Xbf}{{\bm   X}} 
\newcommand{\Lambdabf}{{\mathbf \Lambda}} 
\newcommand{\Sigmabf}{{\mathbf \Sigma}} 

\newcommand{\xbf}{{\bm x}}
\newcommand{\ubf}{{\bm u}}
\newcommand{\wbf}{{\bm w}}
\newcommand{\ybf}{{\bm y}}
\newcommand{\zbf}{{\bm z}}
\newcommand{\zerobf}{{\mathbf 0}}
\newcommand{\onebf}{{\mathbf 1}}
\newcommand{\greekbold}[1]{\mbox{\boldmath $#1$}}

\newcommand{\betabf}{\greekbold{\beta}}
\newcommand{\gammabf}{\greekbold{\gamma}}

\newcommand{\mubf}{\greekbold{\mu}}

\newcommand{\ind}{\mathbbm{1}} % indicator function
\newcommand{\s}{{\sigma^2}}
\newcommand{\R}{{R^2}}

\begin{document}

\def\spacingset#1{\renewcommand{\baselinestretch}%
{#1}\small\normalsize} \spacingset{1}

% \maketitle
% \setstretch{1.75} % 1.75 = 25 lines per page
% \setlength{\parskip}{0cm}
% \setlength{\abovedisplayskip}{3pt}
% \setlength{\belowdisplayskip}{3pt}

%%%%%%%%%%%%%%%%%%%%%%%%%%%%%%%%%%%%%%%%%%%%%%%%%%%%%%%%%%%%%%%%%%%%%%%%%%%%%%
%%% jasa template stuff:

{
  \title{\bf Bayesian Regression Using a Prior on the Model Fit: The R2-D2 Shrinkage Prior}
  \author{{\normalsize Yan Dora Zhang$^{1}$, Brian P. Naughton$^2$, Howard D. Bondell$^{3}$, and Brian J. Reich$^2$}\\ 
  	    \\ {\normalsize $^1$\textit{Department of Statistics and Actuarial Science, The University of Hong Kong }} \\
     {\normalsize $^2$\textit{Department of Statistics, North Carolina State University} }\\
   {\normalsize  $^3$\textit{School of Mathematics and Statistics, University of Melbourne}}}
  \maketitle
}

\bigskip
%%%%%%%%%%%%%%%%%%%%%%%%%%%%%%%%%%%%%%%%%%%%%%%%%%%%%%%%%%%%%%%%%%%%%%%%%%%%%%%
%%% Abstract  %%%%%%%%%%%%%%%%%%%%%%%%%%%%%%%%%%%%%%%%%%%%%%%%%%%%%%%%%%%%%%%%% 
%%%%%%%%%%%%%%%%%%%%%%%%%%%%%%%%%%%%%%%%%%%%%%%%%%%%%%%%%%%%%%%%%%%%%%%%%%%%%%%
\begin{abstract}
Prior distributions for high-dimensional linear regression require specifying a joint distribution for the unobserved regression coefficients, which is inherently difficult. We instead propose  a new  class of shrinkage priors for linear regression via specifying a prior first on the model fit, in particular, the coefficient of determination, and then distributing through to the coefficients in a novel way. The  proposed    method  compares favorably to  previous approaches in terms of both concentration around the origin and tail behavior, which leads to improved performance both in posterior contraction and in empirical performance.  The limiting behavior of the proposed prior is $1/x$, both around the origin and in the tails. This behavior is optimal in the sense that it simultaneously lies on the boundary of being an improper prior both in the tails and around the origin. None of the existing shrinkage priors obtain this behavior in both regions simultaneously. We also demonstrate  that  our proposed  prior  leads to the same near-minimax posterior contraction 
rate as the  spike-and-slab prior. 
\end{abstract}

\noindent%
%{\it Keywords:}  3 to 6 keywords, that do not appear in the title
\vfill

\newpage
\spacingset{1.45} % DON'T change the spacing!
%%%%%%%%%%%%%%%%%%%%%%%%%%%%%%%%%%%%%%%%%%%%%%%%%%%%%%%%%%%%%%%%%%%%%%%%%%%%%%%

% {\it Keywords: Linear Regression, Elicitation, Informative Priors, High-dimensional}

%%%%%%%%%%%%%%%%%%%%%%%%%%%%%%%%%%%%%%%%%%%%%%%%%%%%%%%%%%%%%%%%%%%%%%%%%%%%%%%
%%% Intro  %%%%%%%%%%%%%%%%%%%%%%%%%%%%%%%%%%%%%%%%%%%%%%%%%%%%%%%%%%%%%%%%% 
%%%%%%%%%%%%%%%%%%%%%%%%%%%%%%%%%%%%%%%%%%%%%%%%%%%%%%%%%%%%%%%%%%%%%%%%%%%%%%%
\clearpage
\section{Introduction}
\label{sec:intro}

Consider the linear regression model, 
\begin{equation} 
\label{eq_linear_regression}
Y_i=  \bm{x}_i^T \bm{\beta} + \varepsilon_i,   \ i=1,\cdots, n, 
\end{equation}
where $Y_i$ is the $i$th response,  $ \bm{x}_i$  is the $p$-dimensional  vector of     covariates for the $i$th observation,  $ \bm\beta=(\beta_{1},\cdots,\beta_p)^T$ is the   coefficient vector, and the  $\varepsilon_i$'s are the error terms assumed be normal and independent with  E$( \varepsilon_i) = 0$ and $\text{var}( \varepsilon_i) = \sigma^2$.  
High-dimensional data with $p>n$ in this context   is   common in diverse  application areas. 
It is  well known that   maximum likelihood estimation performs poorly in this setting,  and this motivates a number  of approaches in shrinkage estimation and variable selection.    In   the Bayesian framework,   there  are  two  main  approaches to address such problems:   two component discrete mixture prior (also referred as spike and slab prior)  and  continuous shrinkage priors.   
The  discrete mixture priors  \citep{mitchell1988bayesian, george1993variable,  ishwaran2005spike,  narisetty2014bayesian}   put  a point mass (spike) at $\beta_j=0$ and a continuous prior (slab)  for the terms with $\beta_j\neq0$. Although these priors 
have an intuitive and appealing representation, they   lead to computational issues due to  the  spread of posterior probability over the $2^p$ models formed by including subsets of the coefficients to zero.    
Implementation instead can proceed instead by applying approximation methods, such as stochastic search variable selection \citep{george1993variable},  shotgun stochastic search \citep{hans2007shotgun}, variational Bayes \citep{ormerod2017vb}, and EM \citep{rovckova2014emvs} all of which have improved the computational feasibility and include theoretical underpinnings.

The computation issues with discrete mixture priors motivate  continuous shrinkage priors. The shrinkage priors are  essentially written as global-local    scale  mixture Gaussian  family as summarized in \cite{polson2010shrink}, i.e., 
\[
\beta_j \mid \phi_j, \omega \sim N(0, \omega\phi_j), \ \phi_j\sim\pi(\phi_j), \ (\omega, \sigma^2)\sim \pi(\omega, \sigma^2), 
\]  
where $\omega$ represents the global shrinkage, while $\phi_j$'s are the local variance components. 
Current existing global-local priors  exhibit  desirable  theoretic and  empirical  properties. They can shrink the overall signal, while varying the amount of shrinkage on different components. These continuous priors exhibit both heavy tails and high concentration around zero. The heavy tail reduces the bias in estimation of large coefficients, while the high concentration around zero shrinks the irrelevant coefficients heavily to zero, thus reducing the noise. 
Some examples include 
Normal-Gamma mixtures \citep{griffin2010inference},    Horseshoe    \citep{carvalho2009handling, carvalho2010horseshoe},
generalized Beta \citep{armagan2011generalized}, 
generalized double Pareto \citep{armagan2013generalized}, 
Dirichlet-Laplace   \citep{bhattacharya2015dirichlet},      Horseshoe+  \citep{bhadra2016horseshoe+}, normal-beta prime prior \citep{bai2019large}. 
%\cite{bondell2012consistent} also  proposed a variable selection method based on the penalized posterior credible regions relying on continuous  priors.  
%\textit{Moreover,  several frequentist regularization methods lend themselves to the global-local shrinkage priors, such as  Bayesian lasso \citep{park2008bayesian, hans2009bayesian},  
%Bayesian elastic net \citep{li2010bayesian},  Bayesian bridge \citep{polson2013bayesian}, and Bayesian adaptive lasso \citep{leng2014bayesian}.  }
%Global-local priors can have computational advantages over the discrete mixture priors. 
%Besides these two types of priors,  there is also an increasing literature in applying non-local prior on high-dimensional problems,  such as  \cite{rossell2015non}. 

In general, it is difficult to specify a $p$-dimensional prior on $ \bm{\beta}$, particularly with high dimensional data. Instead, we propose to first construct a prior on the coefficient of determination, $R^2$, for which the model-based version is  defined as the square of the correlation coefficient between the dependent variable  and its modeled expectation. A prior on this one-number summary forms a prior on a function of the parameter vector, and is then distributed through to the individual parameters in a natural way. We develop a class of priors that are constructed via marginalizing over the design, as well as those conditioning on the design. By viewing things in this framework, our proposed class of priors are induced by a Beta$(a,b)$ prior on $R^2$ and lead to priors having desirable properties both asymptotically and in finite samples.

We show that our class of priors, which we term the $R^2$-induced Dirichlet Decomposition   (R2-D2) priors, simultaneously obtain both heavier tails and tighter concentration around zero than all previously proposed approaches. This optimal result translates into improved performance in estimation and inference. We also  offer a theoretical framework to compare different global-local priors.  
 The  proposed    method  compares favorably to  the other global-local shrinkage priors 
 in terms of both its concentration around the origin and its tail behavior obtaining a limiting behavior of $1/x$ in both regions. This behavior is optimal in the sense that it simultaneously lies on the boundary of being an improper prior in both areas, and translates into improved theoretical and empirical performance.

 The rest of the paper is outlined as follows. Section 2 motivates the idea of inducing a prior via $R^2$, and distinguishes between a marginal and conditional version. Section 3 presents the details of the conditional version in both the low- and high-dimensional settings. Section 4 details the marginal version and provides theoretical properties of both the prior and the  posterior. Section 5 discusses novel MCMC algorithms for computation of both the conditional and marginal versions, while Section 6 provides simulation results. Section 7  provides real data examples.
All proofs are given in the Appendix.

%%% Induced model %%%%%%%%%%%%%%%%%%%%%%%%%%%%%%%%%%%%%%%%%%%%%%%%%%%%%%%%%%%%%
\section{Motivation}
%-------------------------------------------------------------------------------

The typical Bayesian approach specifies a joint distribution on the model parameters, namely for the regression coefficients and error variance. 
Instead, we specify a distribution for $R^2$ with practical meaning,   and  then induce  a prior on the $p$-dimensional $ \bm{\beta}$.  

% The  population correlation coefficient $R^2$  between two random variables $Y$ and  $Z$  is defined as
% \[
% \rho  = \frac{\text{cov}(Y,Z)}{\sqrt{\text{var}(Y)\text{var}(Z)}}, 
% \]
%the square of which, or $R$ squared denoted as $R^2$ is then 
%\[
%R^2 = \frac{\text{cov}^2(Y,Z)}{ \text{var}(Y)\text{var}(Z)}. 
%\]

Suppose that  the predictor vector for each observation $ \bm{x} \sim H(\cdot)$,    with E$(  \bm{x}) =  \bm{\mu}$  and $\text{cov}( \bm{x}) = \Sigma$.   
Assume  that $ \bm{x}$ is  independent of the   error, $ \varepsilon$,  and      then the marginal  variance of  $y=\xbf^T\betabf+\varepsilon$ is $\text{var}( \bm{x}^T \bm{\beta})+\sigma^2$. 
For simplicity,  we  assume that the response is centered and  covariates are standardized so that $\bm{\mu} = \bm{0}$,  there is no intercept term in (\ref{eq_linear_regression}), and    all diagonal elements of $\Sigma$ are $1$. 
The coefficient of determination,  $R^2$,    can be calculated as the square of the correlation coefficient between the  dependent variable, $y$,   and  the modeled  value, $ \bm{x}^T \bm{\beta}$, i.e., 
% In this case, $R^2$  measures   how good a  predictor might be constructed from the modeled values,  
\begin{eqnarray} \label{eq_rsquare}
	R^2 = \frac{\text{cov}^2(y, \bm{x}^T \bm{\beta})}{ \text{var}(y)\text{var}( \bm{x}^T \bm{\beta})}
	= \frac{\text{cov}^2( \bm{x}^T \bm{\beta}+\varepsilon, \bm{x}^T \bm{\beta})}{ \text{var}( \bm{x}^T \bm{\beta}+\varepsilon)\text{var}( \bm{x}^T \bm{\beta})}
	=  \frac{\text{var}({ \bm{x}}^T\bm{\beta})}{\text{var}({\bm  x}^T \bm\beta) + \sigma^2} . 
\end{eqnarray}

A hypothesized value of $\R$ has been used previously to tune informative priors, and to select hyper-parameters for regularization problems. 
\cite{scott2014predicting} elicit an informative distribution for the error variance, $\sigma^2$, based on elicitation of the expected $\R$, and the response. 
\cite{zhang2018variable} proposed to choose hyper-parameters for shrinkage priors by empirically minimizing the Kullback-Leibler divergence between the expected distribution of $\R$ and a Beta distribution. Here, in contrast, we develop our approach from first principles via placing a prior distribution on $\R$ directly, rather than using a hypothesized value as a tool to tune parameters in already existing priors. 

Based on this representation of $R^2$, two alternative approaches can be taken to construction of the prior. A conditional version places a Beta prior on the conditional distribution of $\R$ which depends on the model design, $\X$. Conversely, a marginal version assumes that the marginal distribution of $\R$ (after integrating out $\betabf$ and $\X$) has a Beta distribution.

The former has the interpretation of the usual sample-based version of $\R$, while the latter allows for more direct asymptotic analysis of the posterior, as the design is integrated out. We will show that both versions lead to priors having different, but desirable properties. 

\section{Conditional $\R$ Prior} \label{section:conditional prior}
\subsection{$R^2$ as Elliptical Contours}

We now introduce the conditional version, which, conditioning on the design points, yields  
\begin{equation} \label{eq:R2_b}
R^2(\betabf) = \frac{\betabf^T \Xbf^T\Xbf \betabf}{\betabf^T \Xbf^T\Xbf \betabf + n\sigma^2}.
\end{equation}
We specifically write $\R(\betabf)$ to reflect the fact that $\R$ depends on the unknown vector $\betabf$ (as well as $\sigma^2$). 
Notice that (\ref{eq:R2_b}) will reduce to the familiar sample statistic, $R^2$, if the least-squares estimates were substituted for $\betabf$ and $\sigma^2$. 
Conditional on $\sigma^2$ and $\Xbf$, a distribution on $R^2$ induces a distribution on the quadratic form, $\betabf^T\Xbf^T\Xbf\betabf$.

We choose a Beta($a,b$) prior for $R^2$, where the choices of shape parameters $a$ and $b$ will  determine the posterior behavior and will be discussed in more detail in the theoretical results and the implementation. 
An Inverse-Gamma($a_1, b_1$) prior is used for $\sigma^2$, but we note that other choices may also  be  applied. 
A prior for $\betabf$ given $(R^2, \sigma^2)$ then must be defined on the surface of the ellipsoid: $\left\{ \betabf: \betabf^T\X^T\X\betabf = {n R^2 \sigma^2}/({1-R^2}) \right\}$.
When $\X^T\X$ is full rank we may choose $\betabf$ to be uniformly distributed on this ellipsoid;
that is, the distribution of $\betabf$ is constant given the quadratic form. % and $R^2$ induces the  elliptical distribution given in the following proposition.\\
We call this choice the ``uniform-on-ellipsoid'' prior for $\betabf$, and show a connection to a variation on a mixture of $g$-priors.
The following proposition shows that $\betabf$ given $\sigma^2$ has an elliptical distribution after integrating out $\R$.

\begin{proposition}\label{prop1}
    If $R^2 \mid \sigma^2$ has a Beta(a,b) distribution and $\betabf\mid R^2$ has a uniform prior on the ellipsoid, then $\betabf\mid \sigma^2$ has the probability density function:
\begin{equation}
    \label{eq:priorbetabf}
    p(\betabf \mid \sigma^2) = \frac{\Gamma\left( p/2\right)|\Sigmabf_\X|^{1/2} }{B(a,b)~\pi^{p/2}} \left( \sigma^2 \right)^{-a} \left( \betabf^T \Sigmabf_\X \betabf \right)^{a-p/2} \left(1 + \betabf^T \Sigmabf_\X \betabf /\sigma^2\right)^{-(a+b)}, 
\end{equation}
where $\betabf \in \mathbbm R^p, \Sigmabf_\X = \X^T\X/n, $ and $B(\cdot, \cdot)$ denotes the Beta function.
\end{proposition}

As a special case, if $a=p/2$ and $b=1/2$, then $\betabf$ has a multivariate Cauchy distribution with spread parameter $\Sigmabf_\X/\s$. 
\cite{zellner1980posterior} recommended these Cauchy priors for model selection problems.
The next proposition shows that for  $a \leq p/2$, the distribution in (\ref{eq:priorbetabf}) is equivalent to a mixture of normals $g$-prior, with a hyperprior on $g$ that is the product of Beta and Inverse-Gamma distributions.  

\begin{proposition}\label{gibbs.uoe}
If $\betabf\mid \sigma^2, z,w \sim N_p\left(\zerobf, \, zw  \sigma^2 (\X^T\X)^{-1} \right), z \sim $ Inverse-Gamma $(b, n/2), $ $w \sim $ Beta $(a, p/2-a)$, and $a \leq p/2 $, then $\betabf \mid  \sigma^2$ has the distribution given by the density in (\ref{eq:priorbetabf}).  
\end{proposition}
This representation eases the posterior computations for a Gibbs sampler discussed in Section \ref{post.comp}.

%%% Sparse regression and local shrinkage %%%%%%%%%%%%%%%%%%%%%%%%%%%%%%%%%%%%%
\subsection{Sparse Regression and Local Shrinkage}
The prior on $R^2$ regulates $\betabf$ through the quadratic form $\betabf^T\Sigmabf_\X\betabf$, which can shrink the regression coefficients globally, but lacks the flexibility to handle different forms of sparsity.  
In addition, the posterior is not a proper distribution when $\X$ is not full rank (e.g. when $ n < p $). %, then additional restrictions must be imposed to make the posterior proper
%{\color{red} (Is this true? Does improper prior + improper likelihood = improper posterior?)}.
Rather than letting $\betabf \mid (R^2,\s)$ be uniformly distributed on the ellipsoid, we put a Normal-Gamma prior on $\betabf $ \citep{griffin2010inference}, but restrict its support to lie on the surface of the ellipsoid. 
Specifically, we let
\begin{flalign}
    \label{eq:sparseprior}
    \betabf\mid (R^2, \sigma^2, \Lambdabf) &\sim N_p\left( \zerobf, \frac{\s R^2}{1-R^2} \Lambdabf \right) \ind \left\{ \betabf' \Sigmabf_\X \betabf = \frac{\s R^2 }{1-R^2} \right\} \\
    \lambda_j &\sim \text{Gamma} \left( \nu, \mu \right), \text{ for } j = 1, \dots, p,  
\end{flalign}
where  Gamma$( \nu, \mu )$ represents Gamma distribution with shape parameter  $\nu$ and rate parameter $\mu$, 
$\Lambdabf = diag\left\{ \lambda_1, \dots, \lambda_p \right\}$ and $\ind \left\{ \cdot \right\}$ is the indicator function.
Note that this prior no longer requires $\Sigmabf_\X$ to be full rank for it to be proper. 
% The hyper-parameters, $\nu$ and $\mu$, can be chosen to favor sparsity and will be discussed in more detail in the next section.
% This still assumes $\X$ is full rank in order for the ellipsoid to be a bounded subspace, but regularization methods are routinely required for high dimensional data in which $p>n$. 
Proposition \ref{prop.uoe.equiv} shows that the induced model described in the previous section is a special case of the hierarchical model proposed here with fixed $\Lambdabf$.
\begin{proposition}\label{prop.uoe.equiv}
    If $\betabf\mid  (R^2, \s, \Lambdabf) $ has the distribution in (\ref{eq:sparseprior}), and $\Lambdabf = (\X^T\X)^{-1}$, then $\betabf\mid\s$ has the distribution in (\ref{eq:priorbetabf}).    
\end{proposition}
That is, if the contours of the Normal distribution align with the ellipsoid, then we recover the uniform-on-ellipsoid prior.

In general, the conditional distribution of $\betabf$ is similar to a Bingham distribution, which is a multivariate Normal distribution conditioned to lie on the unit sphere. 
This Bingham distribution has density $f(\wbf\mid  \A) = C_A^{-1} \exp\{ -\wbf^T\A\wbf\}$ with respect to the uniform measure on the $p-1$ dimensional unit sphere, where $C_A$ is the normalizing constant \citep{bingham1974antipodally}.
Here, $\betabf\mid  (R^2, \s, \Lambdabf)$ is a Bingham distributed random vector that has been rotated and scaled to lie on the ellipsoid rather than the unit sphere. 
The matrix $\X$ determines the rotation of the ellipsoid, $R^2$ and $\sigma^2$ determine the size of the ellipsoid, and the conditional prior on $\betabf$ determines the direction to the surface.  
If the local variance components ($\lambda_j$) are small, then regions of the ellipsoid near the axes will be favored, encouraging sparser estimates.
Like the Normal-Gamma priors, this is primarily controlled by the shape parameter, $\nu$.
Figure \ref{fig:ellipsoids} illustrates the local shrinkage properties of the prior, showing 10,000 samples of $\betabf \mid  (R^2, \sigma^2)$.
Default choices of the hyper-parameters $\nu$ and $\mu$ follow from the recommendations of \cite{griffin2010inference} and are discussed in the implementation in Section \ref{sim.section}. 

\begin{figure}[htbp]
	    \spacingset{1.1}% DON'T change the spacing!
   \centering
   \includegraphics[width=0.32\linewidth]{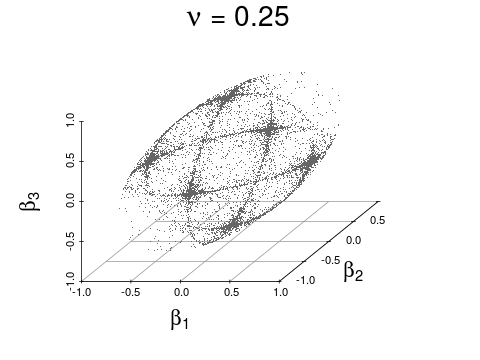}  \hfill
   \includegraphics[width=0.32\linewidth]{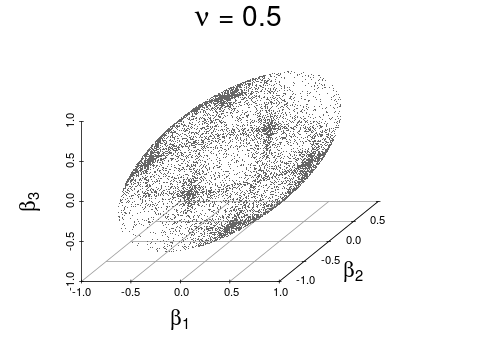}  \hfill
   \includegraphics[width=0.32\linewidth]{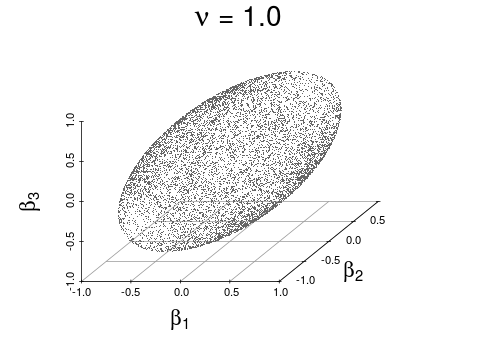}  
   \caption{10,000 samples from the prior of $\betabf\mid (\R, \sigma^2)$ for different choices of $\nu$.}
   \label{fig:ellipsoids}
\end{figure}

%%%%%%%%%%%%%%%%%%%%%%%%%%%%%%%%%%%%%%%%%%%%%%%%%%%%%%%%%%%%%%%%%%%%%%%%%%%%%%%%
%%% Computation %%%%%%%%%%%%%%%%%%%%%%%%%%%%%%%%%%%%%%%%%%%%%%%%%%%%%%%%%%%%%%%
%%%%%%%%%%%%%%%%%%%%%%%%%%%%%%%%%%%%%%%%%%%%%%%%%%%%%%%%%%%%%%%%%%%%%%%%%%%%%%%%
\section{Marginal $\R$ Prior} \label{section:marginal R2-D2}
\subsection{The R2-D2 Global-Local Shrinkage Prior}

Rather than conditioning on the design $\X$, we now instead show how to construct a prior while marginalizing out both $\bm{\beta}$ and the design. While the conditional version retains the interpretation of $R^2$ as elliptical contours in the design space, the marginal version allows for an in depth study of the asymptotic properties of   both the prior and the resulting posterior. 

Consider a prior for   $ \bm{\beta}$  satisfying  E$(\bm{\beta}) = \bm 0$ and $\text{cov}(\bm{\beta}) =  \sigma^2\Lambdabf$, where $\Lambdabf$ is a diagonal matrix with diagonal elements  $ \lambda_1, \cdots, \lambda_p$.  
%Our prior is based on the   $R^2$ which averages over the covariate distribution $H$. To study its prior, we further investigate over the prior for the regression coefficients. 
%Taking the variance of the random function $ \bm{x}^T \bm{\beta}$,    
Then
\begin{eqnarray*}
	\text{var}( \bm{x}^T\bm{\beta}) &=& \text{E}_{\bm{x}}\{ \text{var}_{\bm{\beta}}( \bm{x}^T\bm{\beta} \mid  \bm{x})\} + \text{var}_{\bm{x}}\{ \text{E}_{\bm{\beta}}( \bm{x}^T\bm{\beta}\mid \bm{x})\}  
	= \text{E}_ {\bm{x}} (  \sigma^2  {\bm{x}}^T  \Lambdabf  {\bm{x}} ) + \text{var}_ {\bm{x}}( 0)\\
	&=& \sigma^2  \text{E}_ {\bm{x}}\{ \text{tr} ( {\bm{x}}^T  \Lambdabf   {\bm{x}}  )\}  
	%	= \sigma^2  E_ {\bm{x}}\{ \text{tr} ( \Lambdabf  {\bm{x}}  {\bm{x}}^T  )\}  
	%	=  \sigma^2 \text{tr} ( E_ {\bm{x}}\{   \Lambdabf  {\bm{x}}  {\bm{x}}^T  \})  \\
	= \sigma^2 \text{tr}\{  \Lambdabf E_{\bm{x}} (     {\bm{x}}  {\bm{x}}^T ) \}
	= \sigma^2 \text{tr} (  \Lambdabf  \Sigma )   
	= \sigma^2 \sum_{j=1}^p \lambda_j.
\end{eqnarray*}
Then  $R^2$ is represented as 
\begin{equation} \label{eq_rsquare linear regression model}
R^2  = \frac{\text{var}({{\bm{x}}}^T\bm{\beta})}{\text{var}({ {\bm{x}}}^T\bm{\beta}) + \sigma^2}  =  \frac{\sigma^2 \sum\limits_{j=1}^p \lambda_j}{\sigma^2 \sum\limits_{j=1}^p \lambda_j  + \sigma^2} 
= \frac{  \sum\limits_{j=1}^p \lambda_j}{  \sum\limits_{j=1}^p \lambda_j  + 1}  \equiv \frac{W}{W+1}, 
\end{equation}
where $W \equiv \sum_{j=1}^p \lambda_j $ is the sum of the prior variances scaled by $\sigma^2$.

Similarly as conditional $R^2$ prior, we   also assume  $R^2\sim \text{Beta}(a,b)$,  a Beta distribution with shape parameters $a$ and $b$. Then in this case, 
the induced prior density   for  $ W =  {R^2}/(1-R^2)$  is a Beta Prime distribution \citep{johnson1995continuous} denoted as BP$(a,b)$,  with probability density function 
\begin{equation*}\label{eq_varienceterm's pdf}
\pi_W(x)=  \frac{\Gamma(a+b)}{\Gamma(a)\Gamma(b)}\frac{x^{a-1}}{(1+x)^{a+b}}, \ ( x> 0). %\frac{1}{B(a,b)}\frac{x^{a-1}}{(1+x)^{a+b}} =
\end{equation*}
Therefore $W\sim \text{BP}(a,b)  $ is equivalent to  $R^2\sim\text{Beta}(a,b)$.  The following section will induce a prior on $ \bm{\beta}$ based on the distribution of  the sum of prior variances,  $W$.

Any prior of the form  $\text{E}(\bm \beta) = 0$, $\text{cov}( \bm\beta) =  \sigma^2\Lambdabf$ and $ W = \sum_{j=1}^p \lambda_j \sim \text{BP}(a,b)$ induces a Beta$(a,b)$ prior on $R^2$.  To construct a prior with such properties,   we follow the global-local prior framework and express $\lambda_j =  \phi_j\omega$  with  $\sum_{j=1}^p \phi_j = 1$.  Then $W = \sum_{j=1}^p    \phi_j\omega= \omega$ is the total prior variability, and $\phi_j$ is the proportion of total variance allocated to the $j$-th covariate. It is natural to assume that  $\omega\sim\text{BP}(a,b)$ and the  variances across covariates have a Dirichlet prior with concentration parameter  $(a_\pi,\cdots,a_\pi)$, i.e.,  $ {\phi} = (\phi_1,\cdots, \phi_p)\sim\text{Dir}(a_\pi,\cdots,a_\pi)$. 
Since  
$\sum_{j=1}^p\phi_j = 1$,  
$
\text{E} (\phi_j) = 1/p$, and $ \text{var}(\phi_j) = (p-1)/\{p^2(pa_\pi+1)\}$, 
then smaller $a_\pi$ would lead to  larger variance of ${\phi}_j$,  $j=1,\cdots,p$, thus more  ${\phi}_j $ would be  close to zero with  only a small proportion  of larger components;  
while larger $a_\pi$  would lead to smaller variance of  ${\phi}_j$,  $j=1,\cdots,p$, thus producing a  more 
uniform $\bm{\phi}$, i.e., $\bm{\phi} \approx (1/p, \cdots, 1/p)$. So $a_\pi$ controls the sparsity of the model.

To fully define the global-local prior, we   further  need to  assign a kernel   distribution 
on  each dimension of $ \bm{\beta}$.   
%with  $K(\delta)$ denotes a   kernel (density) with mean zero and variance $\delta$.  
Since the Laplace  distribution ensures more mass around zero and heavier tails than the normal kernel,   we consider a Laplace prior on $\beta_j$ for $j=1,\cdots,p$. 
%
%$\beta_j\mid\sigma^2, \phi_j, \omega \sim   \text{DE}( \sigma (\phi_j \omega/2)^{1/2} )$ for  $j=1,\cdots, p$, 
The prior is then  summarized as 
\begin{equation}\label{eq_R2-D2 prior_dirichlet BP form}
\beta_j  \mid \sigma^2,   \phi_j, \omega \sim \text{DE}( \sigma (\phi_j \omega/2)^{1/2} ), \ 
{\phi} \sim \text{Dir}(a_\pi,\cdots,a_\pi), \ 
\omega \sim \text{BP}(a,b), 
\end{equation}
where 
$\text{DE}(\delta)$  denotes  a double-exponential distribution (i.e., Laplace distribution) with mean $0$ and   variance $2\delta^2$.  
Such prior is induced by a  prior on $R^2$   and the total prior  variance of  $ \bm{\beta}$ is 
decomposed through a Dirichlet prior, 
therefore we refer to the prior as  the  $R^2$-induced Dirichlet Decomposition (R2-D2) prior. 
Here $\omega$ controls the global shrinkage degree through $a$ and $b$, 
while $\phi_j$ controls the local shrinkage through $a_\pi$.  Assume the variance $\sigma^2\sim \text{Inverse-Gamma} (a_1,b_1)$,  an inverse Gamma distribution with shape and scale parameters $a_1$ and $b_1$ respectively.  
%In particular, when $a_\pi$ is small, the prior
%would lead to large variability between the proportions $\phi_j$'s, thus more shrinkage for the regression coefficients; while when $a_\pi$ is large, less shrinkage is assumed. 

\begin{proposition}\label{proposition_1}
	If   $\omega\mid\xi \sim \text{Ga}(a, \xi) $ and $\xi\sim \text{Ga}(b, 1)$, then 
	$\omega\sim\text{BP}(a,b)$, where $\text{Ga}(\mu, \nu)$ is the Gamma random variable with shape $\mu$ and rate $\nu$. 
\end{proposition}

By applying above Proposition \ref{proposition_1},  the prior in (\ref{eq_R2-D2 prior_dirichlet BP form}) can also be written as 
\[
\beta_j\mid\sigma^2,     \phi_j, \omega  \sim    \text{DE}( \sigma (\phi_j \omega/2)^{1/2} ),  \ 
{\phi} \sim \text{Dir}(a_\pi,\cdots,a_\pi), \ 
\omega\mid\xi \sim \text{Ga}(a, \xi), \  \xi\sim \text{Ga}(b, 1). 
\]
%As shown in the next section,    Proposition \ref{proposition_1}'s representation of a    Beta prime variable in terms of  two Gamma variables  reveals   connections among other common shrinkage priors.  

\begin{proposition} \label{proposition_tauphi_j}
	If $\omega \mid \xi \sim \text{Ga}(a, \xi)$,  $(\phi_1,\cdots, \phi_p)\sim\text{Dir}(a_\pi,\cdots,a_\pi)$, and  $a = pa_\pi$, then    it follows  $ \phi_j  \omega \mid \xi \sim \text{Ga}(a_\pi,  \xi)$, $j=1,\dots, p$  independently.  
\end{proposition}

Thus,  by Proposition  \ref{proposition_tauphi_j},   when $a=pa_\pi$,  prior in  (\ref{eq_R2-D2 prior_dirichlet BP form}) can also be written as 
\begin{equation}   \label{eq_R2-D2 prior, linear model, simple representation}
\beta_j\mid\sigma^2,   \lambda_j    \sim   \text{DE}(\sigma (\lambda_j/2)^{1/2}),  \ 
\lambda_j \sim \text{BP}(a_\pi, b). 
\end{equation}  or  equivalently 
\begin{equation}   \label{eq_R2-D2 prior, linear model, simple representation1}
\beta_j\mid\sigma^2,   \lambda_j    \sim   \text{DE}(\sigma (\lambda_j/2)^{1/2}),  \ 
\lambda_j  \mid \xi \sim \text{Ga}(a_\pi, \xi), \ 
\xi \sim \text{Ga}(b,1) . 
\end{equation} 
We   set  $a=p a_\pi$ in  the rest of the paper for the R2-D2 prior.

%-------------------------------------------------------------------------------
\subsection{Properties of the R2-D2 Prior}\label{properties.sec}
%-------------------------------------------------------------------------------
In this section,  the marginal density as well as its theoretical properties of  the proposed  R2-D2  prior    are established.  The properties of the Horseshoe   \citep{carvalho2009handling,carvalho2010horseshoe}, Horseshoe+   \citep{bhadra2016horseshoe+}, generalized double Pareto prior   \citep{armagan2013generalized} and Dirichlet-Laplace  prior \citep{bhattacharya2015dirichlet}    are provided   as a comparison. 
Proofs and technical details are given in the Appendix. 

%Note that there is identifiability issue between  $\sigma^2$   and the global shrinkage parameter, and in this section we focus on the properties of  marginal density of the shrinkage priors. 
% (which in this parameterization actually represents the ratio between the signal variance and the error variance), just the same as if it was not included in the prior, since the errors do give some separate information about the scale, $\sigma^2$. 
For simplicity of comparison across different priors, the  variance term $\sigma^2$ is fixed at  $1$.

\begin{proposition} \label{proposition_marginal density}
	Given the  R2-D2  prior (\ref{eq_R2-D2 prior, linear model, simple representation}), the marginal density of $\beta_j$ for any $j=1,\cdots, p$ is  
	\begin{equation*} \label{eq_meijerG marginal density}
	\pi_{\text{R2-D2}}(\beta_j) =  \frac{	 G^{3,1}_{1,3}\left( \frac{\beta_j^2}{2} \left\vert_{a_\pi-\frac{1}{2},0,\frac{1}{2}}^{\frac{1}{2}-b}\right.\right)}{  {(2\pi)^{1/2}}\Gamma(a_\pi)\Gamma(b) } 
	=  \frac{ G^{1,3}_{3,1}\left( \frac{2}{\beta_j^2} \left\vert_{ \frac{1}{2}+b }^{\frac{3}{2} -  a_\pi,1,\frac{1}{2}}\right.
		\right)  }{{(2 \pi)^{1/2}}\Gamma(a_\pi)\Gamma(b) } 
	\end{equation*} where $\Gamma(\cdot)$ denotes the Gamma function and 
	  $ G^{m,n}_{p,q}\left(z \left\vert_{a_1, \dots, a_p}^{b_1,\dots, b_q} \right.  \right)$ denotes the Meijer G-function (see Appendix for the   detailed definition). 
\end{proposition}

Now we would like to compare the theoretical properties with our proposed R2-D2 prior with a couple of common global-local shrinkage priors. We first listed these priors. 

The Horseshoe prior proposed in  \cite{carvalho2009handling,carvalho2010horseshoe} is 
\[
\beta_j \mid\lambda_j \sim N(0, \lambda^2_j),  \text{   } \lambda_j \mid\tau  \sim C^+(0,\tau), 
\]  
where $C^+(0, \tau)$ denotes a half-Cauchy distribution with scale parameter $\tau$, with density 
$\pi(y\mid\tau) = 2/\{\pi\tau (1+(y/\tau)^2)\}$.
 
The Horseshoe+ prior proposed in  \cite{bhadra2016horseshoe+} is 
\[
\beta_j \mid\lambda_j \sim N(0, \lambda^2_j),  \    \lambda_j \mid\tau , \eta_j \sim C^+(0,\tau\eta_j),  \text{   }
\eta_j\sim C^+ (0,1). 
\]
 
The Dirichlet-Laplace prior proposed in \cite{bhattacharya2015dirichlet} is 
\begin{equation}   \label{eq_dl prior}
\beta_j\mid  \psi_j  \sim  \text{DE}(\psi_j),   \ 
\psi_j \sim \text{Ga}(a^\ast,1/2). 
\end{equation}
 
	The normal-beta prime prior proposed in  \cite{bai2019large} is as follows: 
\begin{equation}\label{eq_normal beta prime}
\beta_j \mid \tau_j \sim N(0, \tau_j), \ 
\tau_j \sim \text{BP}(a^\#,b^\#). 
\end{equation}

%\cite{van2017adaptive} further investigate the frequentist properties of Bayesian procedures for estimation based on the horseshoe prior in the sparse  normal means model, leading to near minimax optimal estimation. 
%
%\cite{van2014horseshoe}   show that if the number of nonzero parameters of the mean
%vector is known, the horseshoe estimator attains the minimax $L_2$ risk.  
%
%\cite{van2017uncertainty} investigate the credible sets and marginal credible intervals resulting from Horseshoe prior in the sparse multivariate normal means model, show that credible balls and marginal credible intervals have good frequentist coverage and optimal size if the  sparsity level of the prior is set correctly. 
%In the high-dimensional  linear regression setup, \cite{song2017nearly} show that if the shrinkage prior has a heavy and flat tail, and allocates a suffciently large probability mass in a very small neighboorhood of zero, then its posterior properties are as good as those of the spke-and-slab prior, leading to a nearly optimal contraction rate and selection consistency as the spike-and=slab prior. 

\begin{figure}[h!] 
	    \spacingset{1.1}% DON'T change the spacing!
	\centering
	\includegraphics[scale=0.68]{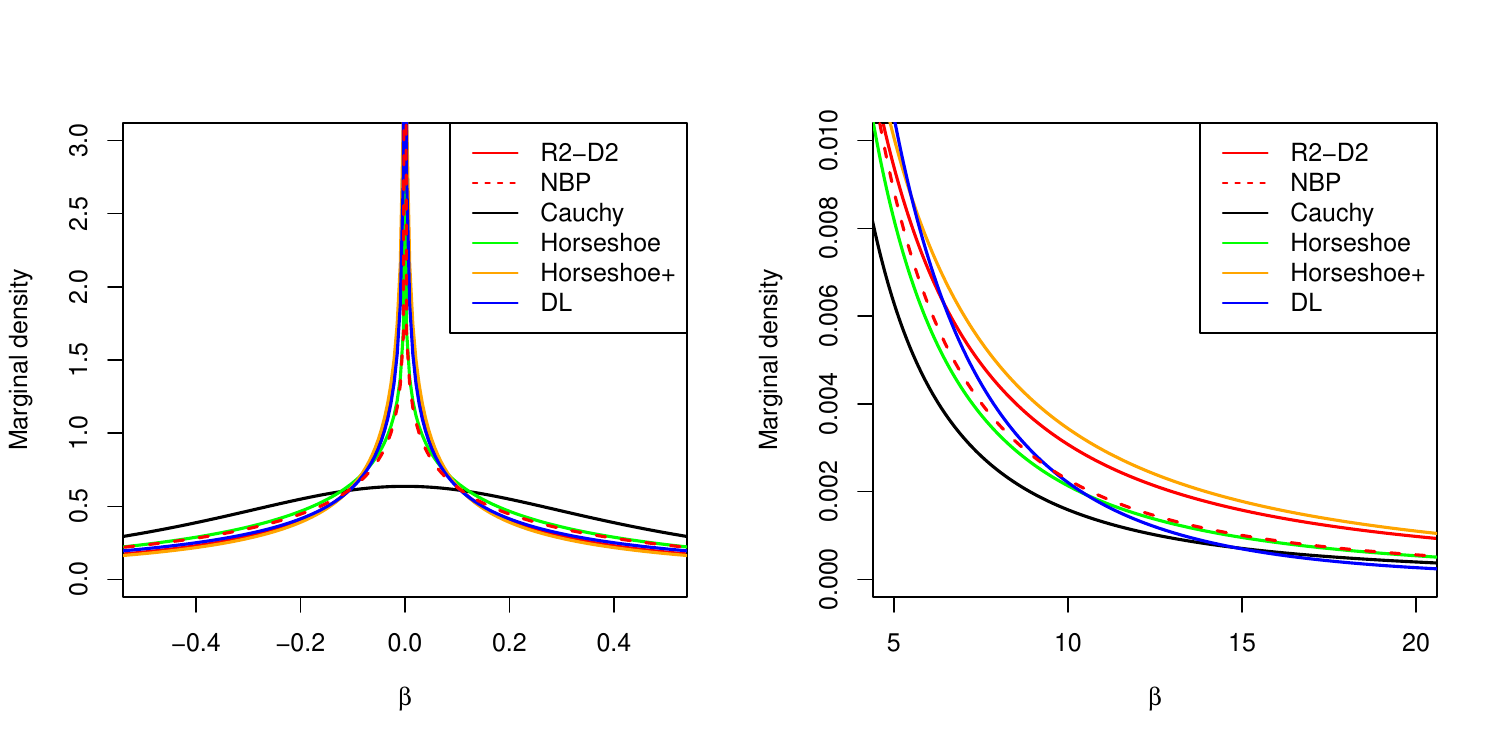}
	\caption{Marginal density of the  R2-D2 (R2-D2),  normal-beta prime (NBP),  Dirichlet-Laplace (DL),  Horseshoe, Horseshoe+ prior and  Cauchy distribution.  
		In all cases, the  hyper-parameters are selected  to ensure  the  interquartile range  is 1 for visual comparison.   }\label{figure_density function}
\end{figure}

Figure \ref{figure_density function} plots   the marginal density function of the  R2-D2 density along with   the  normal-beta prime, Horseshoe, Horseshoe+,  Dirichlet-Laplace,  and Cauchy distributions.   In the figure,  for   visual comparison, the hyper-parameter  in the priors, i.e.,  $\tau$ in  Horseshoe and Horseshoe+ prior, $a^\ast$ in Dirichlet-Laplace prior, $(a_\pi, b)$ in  the  R2-D2 prior, 
are   selected   to ensure the  interquartile range is 1.    Note that to make the plots comparable,   $a_\pi$ in the proposed R2-D2 prior  is set as $a^\ast/2$, which is half of the hyper-parameter  in the  Dirichlet-Laplace prior.
 It will be shown later that this results in the same behavior around the origin for the two priors. 
The other hyper-parameter  $b$   in the R2-D2  prior is then  tuned to ensure the interquartile range   of 1 to match the others.   For the normal-beta prime prior, we follow the values in \cite{bai2019large}, i.e.,  $a^\#=0.48$ and $b^\#=0.52$ which also ensures an interquartile range of 1.

From the plot, it appears that the R2-D2 prior can obtain both the highest concentration around zero and heaviest tail simultaneously. We will quantify these rates exactly in the next subsection, in Table 1. In particular, we will see that the R2-D2 prior is the only one obtaining polynomial behavior in both regions.

In the normal means model,  \cite{van2014horseshoe} and \cite{van2017adaptive} investigate the  Horseshoe  posterior contraction rate,   \cite{bhattacharya2015dirichlet} shows the optimal  posterior concentration results for Dirichlet-Laplace prior, \cite{bhadra2016horseshoe+} proves that the Horseshoe+ posterior concentrates at a faster rate than Horseshoe in the Kullback-Leibler sense, and \cite{bai2019large} shows that normal-beta prime prior leads to a near minimax posterior concentration rate.     

In this paper, we examine the concentration around zero and tail behaviors of the marginal densities of a number of priors, and show that our proposed approach simultaneously achieves high concentration at the origin and heavy tails. We will also study the posterior consistency and contraction properties  in the high-dimensional regression model setup. As shown in  Figure \ref{figure_density function},   all five global-local shrinkage priors have a marginal density with a singularity at zero while with different concentration rate.  Except for the Dirichlet-Laplace prior, all other priors' marginal density have a heavier tail than the Cauchy distribution. We formally investigate their marginal densities' properties in the following sections.   
%From Figure \ref{figure_density function},  the  R2-D2 prior  density  shows 
%the most mass around zero  and the heaviest tails;  

%-------------------------------------------------------------------------------
\subsubsection{Asymptotic tail  behaviors}\label{Section_tail properties}
%-------------------------------------------------------------------------------
We  examine the behavior of the tails of the proposed   R2-D2 prior in this section.  A prior with heavy tails is desirable in high-dimensional regression to allow the posterior to estimate  large values   for important predictors.  
\begin{theorem}\label{theorem_tail properties}
	Given  $|\beta|\rightarrow\infty$, for any  $a_\pi>0$ and $b>0$, the marginal density of the  R2-D2 prior  (\ref{eq_R2-D2 prior, linear model, simple representation})
	satisfies 
	$\pi_{\text{R2-D2}}(\beta) = O({1}/{|\beta|^{2b+1} }) $. 
	Furthermore,  when $0<b<1/2$, $\lim_{|\beta|\rightarrow\infty} {\pi_{\text{R2-D2}}(\beta)}/\beta^{-2} = \infty $, i.e., the  R2-D2 prior has heavier tails than the Cauchy distribution. 
\end{theorem}
With a polynomial tail heavier than Cauchy distribution, the proposed R2-D2  prior attains a substantial improvement over a large class of global-local shrinkage priors.  

As a comparison, we   study the tail behavior of the Dirichlet-Laplace   and double Pareto prior.  The  density of generalized double Pareto prior  proposed in \cite{armagan2013generalized} is 
\[
\pi_{\text{GDP}}  ( \beta_j \mid \eta, \alpha )= ( 1+ |\beta_j| / \eta) ^ {_-(\alpha + 1)}/  (2\eta/\alpha) , \ (\alpha, \eta >0). 
\]

\begin{theorem} \label{theorem_tail properties Pareto}
	Given  $|\beta| \rightarrow \infty$,  for any $\alpha>0$, 
	the marginal density of the   generalized double Pareto prior    satisfies  $\pi_{\text{GDP}}(\beta)  = O( 1/ {|\beta|^{\alpha + 1}})  $.  Furthermore,  when $\alpha < 1$, 
	$\lim_{|\beta|\rightarrow\infty} {\pi_{\text{GDP}}(\beta) }/{\beta^{-2} } = \infty $, i.e., the double Pareto  prior has heavier tails than the Cauchy distribution.  
\end{theorem}

\begin{theorem} \label{theorem_tail properties DL}
	Given  $|\beta| \rightarrow \infty$,  for any $a^\ast>0$, 
	the marginal density of the  Dirichlet-Laplace prior  as shown in (\ref{eq_dl prior})   satisfies  $\pi_{\text{DL}}(\beta)  = O( {|\beta|^{a^\ast/2-3/4}}/{  \exp\{ \sqrt{2|\beta|} \}  })  $.  Furthermore, 
	$\lim_{|\beta|\rightarrow\infty} {\pi_{\text{DL}}(\beta) }/{\beta^{-2} } = 0 $, i.e., the Dirichlet-Laplace  prior has lighter tails than the Cauchy distribution.  
\end{theorem}

%%%%%%%

%%\begin{proposition} \label{theorem_tail properties_NBP}
%	Given  $|\beta|\rightarrow\infty$, for any  $a^\#>0$ and $b^\#>0$, the marginal density of the  normal-beta prime prior  as shown in (\ref{eq_normal beta prime})
%satisfies 
%$\pi_{\text{NBP}}(\beta) =o({1}/{|\beta|^{2b^\#+1} }) $. 
%\end{proposition}
%

As noted in \cite{carvalho2010horseshoe}, the Horseshoe prior has  exact Cauchy-like tails that decay like $\beta^{-2}$, and the Horseshoe+ prior has a tail of $O({\log(|\beta|)}/{\beta^2})$  as illustrated in the proof of Theorem 4$.$6 in \cite{bhadra2016horseshoe+}. 
Therefore, the double Pareto prior and the  proposed   R2-D2  prior lead to   the heaviest tail, followed by  Horseshoe+,  then  Horseshoe, and  finally Dirichlet-Laplace prior.

%-------------------------------------------------------------------------------
\subsubsection{Concentration properties}\label{Section_concentration properties}
%-------------------------------------------------------------------------------
In this section, we study the concentration properties of the R2-D2 prior around the origin. The concentration properties of Dirichlet-Laplace, Horseshoe, and Horseshoe+ priors are also given.  We favor priors with high concentration near zero to reflect the prior that most of the covariates do not have a substantial effect on the response.   We now show that   R2-D2 prior has higher  concentration at zero to go along with   heavier tails than other global-local priors. 
\begin{theorem}\label{theorem_center properties R2-D2}
	As $|\beta|\rightarrow0$,   
	if  $0< a_\pi < {1}/{2}$ and $b>0$,  
	the marginal density of the  R2-D2  prior as shown in  (\ref{eq_R2-D2 prior, linear model, simple representation})
	satisfies 
	$\pi_\text{R2-D2}(\beta)  = O( 1/ |\beta|^{ 1- 2a_\pi})$. 
	%	$\pi_\text{R2-D2}(\beta) \approx  c_1+ \frac{c_2}{\beta^{1-2a_\pi}} \rightarrow\infty$, 
	%	where  $c_1 =  \frac{1}{ (2\pi)^{1/2}\Gamma(a_\pi)\Gamma(b)} \Gamma(a_\pi-\frac{1}{2}) \Gamma(\frac{1}{2}) 
	%	\Gamma(\frac{1}{2}+b)$ and 
	%	$ c_2 =  \frac{1}{2^{a_\pi}\pi^{1/2}\Gamma(a_\pi)\Gamma(b) }\Gamma(\frac{1}{2}-a_\pi)\Gamma(1-a_\pi) \Gamma(a_\pi+b)$ are constant values. 
\end{theorem}

\begin{theorem}\label{theorem_center properties DL}
	As $|\beta|\rightarrow 0 $,   if  $0<a^\ast<1$,   the marginal density of the  Dirichlet-Laplace prior  as shown in (\ref{eq_dl prior}) satisfies    $
	\pi_{DL}(\beta)   = O(1/|\beta|^{1-a^\ast} ) $. 
	%  where $c =  \frac{ \Gamma(1-a^\ast) }{2^{1+a^\ast}\Gamma(a^\ast)} $ is a constant value. 
\end{theorem}

%%%%%%%

%\begin{proposition} \label{theorem_center properties_NBP}
%	Given  $|\beta|\rightarrow0$, for any  $a^\#\in (0,1/2)$ and $b^\#>0$, the marginal density of the  normal-beta prime prior  as shown in (\ref{eq_normal beta prime})
%	satisfies 
%	$\pi_{\text{NBP}}(\beta)^{-1} =o(|\beta|^{1-2a^\#}) $. 
%\end{proposition}

For the	Horseshoe prior,  as summarized in \cite{carvalho2010horseshoe},   the marginal density 
$
\pi_{HS}(\beta) =    {(2\pi^3)^{-1/2}}  \exp(\beta^2/2 )E_1(\beta^2/2),
$    where $E_1(z) = \int_{1}^{\infty} e^{-tz}/t \, dt$ is the exponential integral function. 
As $|\beta|\rightarrow 0$, 
\[
\frac{1}{   2{(2\pi^3)^{1/2}} }  \log(1+\frac{4}{\beta^2}) \leq 
\pi_{HS}(\beta) \leq \frac{1}{  {(2\pi^3)^{1/2}} }  \log(1+\frac{2}{\beta^2}). 
\]
Therefore around the origin, $\pi_{HS}(\beta)  = O(\log ( {1}/{|\beta|}))$. 
Also by the proof of Theorem 4$.$6 in \cite{bhadra2016horseshoe+}, as  $|\beta|\rightarrow 0$,  the marginal density of Horseshoe+ prior satisfies $\pi_{HS+}(\beta)  = O(\log^2(  {1}/{|\beta|}) )$. 

It is clear that  $2a_\pi$ in the R2-D2  prior  plays the same role around the  origin as   $a_D$ in the Dirichlet-Laplace prior. 
Accordingly,   when   $a^\ast=2a_\pi\in(0,1)$,    all  these  four priors   possess unbounded density near the origin. However,  the R2-D2  prior  and Dirichlet-Laplace  prior    diverge 
to infinity  with a   polynomial order, much  faster than the Horseshoe+ (with a squared logarithm order)  and the Horseshoe prior (with a logarithm order).   Although the double Pareto prior also has a polynomial order tail similar as our proposed  R2-D2  prior,  the double Pareto prior   differs around the origin, as it remains bounded, while  our proposed R2-D2 prior is unbounded at the origin.   As for the normal-beta prime prior, we show that its concentration rate is slower than the R2-D2 prior. 

The  results for all priors in both tail behavior and concentration around zero are summarized in  Table \ref{table-asymptotics}. The proposed R2-D2 prior is the only one that can achieve polynomial rates both in the tails as well as around zero. This is the limiting case in that it can then be arbitrarily close to the boundary case of $1/{|\beta|}$ in each region.

\begin{table} [h!]
    \spacingset{1.1}% DON'T change the spacing!
	\renewcommand{\arraystretch}{1.5}
	\begin{center}
		\normalsize\addtolength{\tabcolsep}{-2pt}
		\begin{tabular}{| c| c |c| }
			\hline 
			& Tail Decay & Concentration at zero  		\\
			\hline 
			Horseshoe &  $O\left(\frac{1}{\beta^2}\right)$ &  
			$O\left(\log(\frac{1}{|\beta|})\right)$   \\ \hline
			Horseshoe+ &  $O\left(\frac{\log |\beta|}{\beta^2}\right)$
			& $O\left(\log^2(\frac{1}{|\beta|})\right)$   \\	\hline
			Dirichlet-Laplace & $O\left( \frac{|\beta|^{a^\ast/2-3/4}}{  \exp\{ \sqrt{2|\beta|}\} } \right)$ & $O\left(\frac{1}{|\beta|^{1-a^\ast}}\right)$   \\ \hline
			Generalized Double Pareto &	$O\left( \frac{1}{|\beta|^{1+\alpha} }\right)$ & $O (1)$ \\ \hline
%				Normal-beta Prime &  	$o\left( \frac{1}{|\beta|^{1+2b^\#} }\right)$ &
%			$\succ  \frac{1}{|\beta|^{1-2a^\#}}  $  \\	
			R2-D2 &  	$O\left( \frac{1}{|\beta|^{1+2b} }\right)$ &
			$O\left(\frac{1}{|\beta|^{1-2a_\pi}}\right) $  \\	
			\hline 
		\end{tabular}
	\end{center}
	\caption {Tail behavior and concentration around zero for Horseshoe, Horseshoe+,  Dirichlet-Laplace, Generalized Double Pareto, and R2-D2 priors.} 
\label{table-asymptotics}
\end{table}

%------------------------------------------------------------------------------- 
\subsubsection{Consistency and contraction of R2-D2 Posterior}\label{Section_posterior consistency}
%------------------------------------------------------------------------------- 

%---%---%---%---%---%---%---%---
In what follows, we rewrite   $p$ as $p_n$ to indicate the dimension of $\bm{\beta}_n$ can increase with the sample size $n$. 
We denote the true regression parameter as  $ \bm{\beta}^0_n$.   
Let $q_n$ be the number of nonzero components in   $ \bm\beta^0_n$. 
We use $||\cdot||$ as the $L_2$  norm,   and $||\cdot||_1$ as $L_1$ norm for vectors, respectively. 
Denote $\bm Y_n =  (Y_1, \cdots, Y_n)^T$. 
Denote $\bm X_n = ({{\bm{x}}}_1^T, \cdots, {{\bm{x}}}_n^T)^T$.  
Let $\xi_n$ denote a set of indices where $\xi_n \subset  \{1, \cdots, p_n \}$, and let $\bm{X}_{\xi_n}$ denote the sub-matrix of ${\bm X}_n$ that contains the columns with indices in $\xi_n$. 
Denote $\xi^0_n$  as the set containing  the nonzero indices of $\bm{\beta}^0_n$. 
For two positive sequences $a_n$ and $b_n$, 
$a_n\asymp b_n$ means $0<\liminf a_n/b_n \leq \limsup a_n/b_n<\infty$;
$a_n\prec b_n$ means $a_n = o(b_n)$; 
$a_n\succ b_n$ means $b_n = o(a_n)$; 
$a_n\preceq b_n$ means either $a_n  \prec b_n$ or $a_n \asymp b_n$; 
and $a_n\succeq b_n$ means either $a_n  \succ b_n$ or $a_n \asymp b_n$. 

We now show that  the proposed R2-D2 prior     yields   strong posterior consistency in the case that $p_n \prec n$, and further that it attains the optimal near-minimax contraction rate in general, including with  $p_n\succeq n$, as given by \cite{castillo2015bayesian}, \cite{song2017nearly} and \cite{rovckova2018spike}.
%In Bayesian, posterior consistency means that  
%the entire  posterior  distribution of $ \bm\beta$ converges in probability towards the truth $\bm{\beta}^0$.  

%---%---%---%---%---%---%---%---
Theorem \ref{theorem_R2-D2_consistency} shows strong posterior consistency under $p_n = o(n)$, and assumes the  following regularity conditions: 
\begin{enumerate} [label=(A\arabic*)]  
	\item  \label{a_p=o(n)}    $p_n \prec n$;
	\item  \label{a_eigenvalues} Let $d_{p_n}$ and $d_1$ be the smallest and the largest singular values of $ { \bm{X}_n^T \bm{X}_n}/{n}$ respectively. Assume
	$0<d_{\min}<\liminf_{n\rightarrow\infty} d_{p_n} \leq \limsup_{n\rightarrow\infty}d_1 < d_{\max} <\infty$, where $d_{\min}$ and $d_{\max}$ are fixed; 
	%(Let $\Lambda_{n\min}$ and $\Lambda_{n\max}$ be the smallest and the largest singular values of $X$, respectively. Then $0<\Lambda_{\min}<\liminf_{n\rightarrow\infty}\Lambda_{n\min}/\sqrt{n} \leq \limsup_{n\rightarrow\infty}\Lambda_{n\max} /\sqrt{n} < \Lambda_{\max} <\infty$.)
%	\item  \label{a_beta0_sup} 
%	$\limsup_{n\rightarrow\infty} \max_{j=1,\cdots,p_n}|\beta_{nj}^0| <\infty$;
	\item  \label{a_beta0_sup} $  \max_{j=1,\cdots,p_n} \{ |\beta_{nj}^0 | \} \leq E_n$ for some nondecreasing sequence, $E_n$, with $\log(E_n) = O(\log n)$;
 
	\item \label{a_qplog}  $q_n = o( n/ \log n)$. 
	%	 in which $q_n$ is the number of nonzero components in $ \bm\beta^0$. 
\end{enumerate}
\begin{theorem} \label{theorem_R2-D2_consistency}
	Under assumptions \ref{a_p=o(n)}--\ref{a_qplog},  for any $b>0$, given the linear regression model (\ref{eq_linear_regression}) with known $\sigma^2$,  if  $a_{\pi}= C/( p_n^{b/2} n^{r b/2 }\log n)$  for finite $r>0$ and  $C>0$,  then the  R2-D2 prior   (\ref{eq_R2-D2 prior, linear model, simple representation})   yields a   strongly consistent posterior, i.e., 
	for any $\epsilon>0$,  
	\[
	\text{Pr}_{\bm\beta^0_n}\left\{ \pi_n ( \bm{\beta}_n:|| \bm{\beta}_n- \bm{\beta}_n^0||\geq \epsilon \mid {\bm Y}_n)\rightarrow 0 
 \right\} = 1	\text{  as $n\rightarrow \infty$.}
	\]

\end{theorem}

%---%---%---%---%---%---%---%---
Theorem \ref{theorem_R2-D2_consistency} shows the posterior strong consistency of the R2-D2 prior  under $p_n\prec n$. 
Furthermore,  in the high dimensional case, with $p_n\succeq n$, we place an inverse-Gamma prior on $\sigma^2$, and denote	$\sigma^0$  as the true parameter value  which is unknown but fixed.  
Theorem \ref{theorem_new theorem}  shows that the R2-D2 prior  contracts at the near-minimax rate in this regime, under the following regularity conditions: 
\begin{enumerate} [label=(B\arabic*)]  
		\item \label{new0}  All the covariates are uniformly bounded. For simplicity, we assume they are all bounded by 1; 
	
	\item \label{new1} $p_n\succeq n$;
	\item \label{new3}   There exists some integer $\bar p_n$ and fixed constant $d_0$ such that $\bar p_n\succ q_n  $ and the smallest singular value of matrix $ {\bm X}_{\xi_n}^T {\bm X}_{\xi_n}/n  $ is no smaller than 
	$ d_0$ for any subset model of size $|\xi_n| \leq \bar{p}_n$; 
		\item \label{new5}  $\text{max}_{j=1,\dots,p_n} \{ |\beta_{nj}^0/\sigma^0| \} \leq E_n$ for some nondecreasing sequence, $E_n$, with $\log(E_n) = O(\log p_n)$. 
	\item \label{new4} $q_n = o(n/\log p_n)$;

\end{enumerate}

\begin{theorem}\label{theorem_new theorem}
	Assume that    \ref{new0}-\ref{new4}  hold. Denote $\epsilon_n = M\sqrt{q_n (\log p_n)/n}$  where   $M>0$ is sufficiently large, and let $k_n\asymp  \sqrt{ q_n(\log p_n)/n}/p_n$. 
	Given 
	the linear regression model (\ref{eq_linear_regression}), suppose that we  place an inverse-Gamma prior  on  $\sigma^2$ and 
	place the R2-D2 prior  (\ref{eq_R2-D2 prior, linear model, simple representation})  on $\bm\beta$. 	 
	For any  $b>0$, if 
	$a_{\pi}  \leq   {\log\left(1-p_n^{-1+u} \right)} / ({2\log k_n})$ 
	where $u>0$,  
	then the  following hold: 
	\[
	\text{Pr}_{\bm\beta^0_n}  \Big\{\pi_n (\bm\beta_n: || \bm\beta_n - \bm\beta^0_n || \geq c_1 \sigma^0  \epsilon_n \mid \bm Y_n ) \geq e^{-c_2 n\epsilon_n^2}   \Big\}  \leq e^{-c_3 n\epsilon_n^2}, 
	\]	
	\[
	\text{Pr}_{\bm\beta^0_n}  \Big\{ \pi_n  (\bm\beta_n: || \bm\beta_n - \bm\beta^0_n ||_1 \geq c_1 \sigma^0  \sqrt{q_n}\epsilon_n \mid \bm Y_n ) \geq e^{-c_2 n\epsilon_n^2}   \Big\}   \leq e^{-c_3 n\epsilon_n^2}, 
	\]
	\[
	\text{Pr}_{\bm\beta^0_n}  \Big\{ \pi_n  (\bm\beta_n: || \bm X_n\bm\beta_n - \bm X_n \bm\beta^0 _n|| \geq c_1 \sigma^0   \sqrt{n}\epsilon_n \mid \bm Y_n ) \geq e^{-c_2 n\epsilon_n^2}   \Big\}  \leq e^{-c_3 n\epsilon_n^2}, 
	\]

	%\[
	%\text{Pr}_{\bm\beta^0} \Big\{ \pi (\bm\beta:  \text{at least $\tilde{p}_n$ entries of $|\bm\beta/\sigma_n|$ are larger than  $k_n$} \mid \bm Y )  >  e^{-c_2 n\epsilon_n^2}   \Big\}  \leq e^{-c_3 n\epsilon_n^2}, 
	%\]
	for some  positive constants   $c_1$, $c_2$, and  $c_3$. 
\end{theorem}

According to \cite{raskutti2011minimax}, the minimax $L_2$ rate is $O(\sqrt{q_n \log (p_n/q_n)/n})$. 
For the R2-D2 prior, the $L_2$ and $L_1$ contraction rates for the posterior of $\bm\beta_n$ are  $O(\sqrt{q_n (\log p_n)/n})$ and $O(q_n\sqrt{ (\log p_n)/n})$, respectively.  So the  contraction rates of R2-D2 prior are near-minimax. 
Note that these rates are the same as the rates achieved by spike-and-slab approaches as in  \cite{castillo2015bayesian}, \cite{song2017nearly} and \cite{rovckova2018spike}. 

We note that Theorem 6 shows strong consistency in the $p_n \prec n$ regime, while Theorem 7 shows the contraction rate for $p_n\succeq n$. While, the result in Theorem 7 is stronger, including stronger conditions on the hyperparameters, we do conjecture that the near-minimax contraction rate will also hold in the case of $p_n \prec n$, with a condition on the hyperparameters that is weaker than that of Theorem 7. As also pointed out in \cite{song2017nearly}, in this case it is not necessary to require a strong prior concentration which is ensured by the conditions of Theorem 7, and one only need to impose conditions on the local shape of the prior around the true $\mathbf{\beta}^0$.

\subsubsection{Choice of Hyperparameters}
Based on these properties, we now discuss a default choice of hyperparameters that can be implemented in practice. We set $a_\pi = a/p$ throughout, which is not as a choice, but it is an important step in the definition of the prior. The reason for this very specific relationship is to ensure that the R2-D2 prior in  (\ref{eq_R2-D2 prior_dirichlet BP form}) can be re-written as (\ref{eq_R2-D2 prior, linear model, simple representation}) and (\ref{eq_R2-D2 prior, linear model, simple representation1}). 
The theoretical properties are derived based on (\ref{eq_R2-D2 prior, linear model, simple representation})  (or (\ref{eq_R2-D2 prior, linear model, simple representation1})), hence they assume this specific form of $a_\pi$, as any other choice would no longer yield the same theoretical results. In addition, it ensures that the MCMC algorithm is fully Gibbs sampling. With any other choice of $a_\pi$, this would not be the case, and a Metropolis step would be needed in the algorithm. 

Hence there are then two  parameters to set, $a$ and $b$.  Based on the consistency results, we suggest to set $a$ as a function of $b$ based on the condition of Theorem 6. 
This leads to just one tuning parameter $b$. This determines the tail behavior and then all other parameters are fixed from that. For a fully default method, we set $b = 0.5$ to yield Cauchy-like tails, but other choices of tail behavior are possible if desired. 

To determine $a$ and $a_\pi$ from a choice of $b$, note that the condition for consistency in Theorem 6 requires that $a_\pi=C/(p_n^{b/2}n^{r b/2} \log n)$, where $C$ and $r$ are arbitrary constants. For a default approach, given the choice of $b$, we set $a_\pi$ exactly at $C/(p_n^{b/2}n^{r b/2} \log n)$ choosing the arbitrary constants $C$ and $r$ each to be 1. This is now the default choice and has been implemented in all of the examples to follow.

\section{Posterior Computation}\label{post.comp}
In this section we develop novel Markov chain Monte Carlo (MCMC) samplers for both the conditional and marginal R2-D2 approaches. The development of these samplers are of interest directly on their own, as after some  transformation and reparametrization, we are able to obtain efficient feasible methods. In particular, the marginal version and the conditional uniform on ellipsoid version allow for fully Gibbs samplers. The conditional version with local shrinkage requires a Metropolis-Hastings sampler and we show how to sample from this posterior even in the case where $p > n$ and hence $\X^T\X$ is not full rank. 

%%% Uniform-on-Ellipsoid Model %%%%%%%%%%%%%%%%%%%%%%%%%%%%%%%%%%%%%
\subsection{Gibbs Sampler for the Conditional Uniform-on-Ellipsoid Prior}
In Proposition \ref{gibbs.uoe}, we showed that $\betabf\mid  \sigma^2$ has a mixture of normals representation for $a \leq p/2$. In practice, we recommend choosing $a = p/n$ as a default, and hence this applies as long as $n > 1$. 
If $\sigma^2$ is Inverse-Gamma($a_1,b_1$) distributed, then $\betabf, \sigma^2, z$ and $w$ are drawn from their full conditional distributions in the Gibbs sampler described as follows: 
\begin{enumerate}[label=(\alph*)]
    \item Set initial values for  $\betabf, \sigma^2, z$ and $w$.
    \item Sample $w\mid \betabf, \s, z$, by first sampling $u \sim $ Gamma$\left( p/2 - a, \,  \betabf^T\X^T\X\betabf/\left(  2g\s\right)\right)$, and setting $w = 1/(1+u)$.
    \item Sample $z\mid\betabf, \s, w \sim$  Inverse-Gamma$\left( p/2+b,\, \betabf^T\X^T\X\betabf/\left(2w\s\right)+1/2 \right)$.
    \item Sample $\s \mid\betabf, z, w \sim$  Inverse-Gamma$\left(  (n+p)/2 + a_1, \, \text{SSE}/2 +\betabf^T\X^T\X\betabf/\left(2wz\right) + b_1 \right)$, where SSE $=\left( \ybf -\X\betabf \right)^T\left( \ybf -\X\betabf \right)$.
    \item Sample $\betabf\mid\sigma^2, z,w \sim $ Normal$\left( c \widehat \betabf_{LS}, \,c \s (\X^T\X)^{-1} \right)$, where $ c = \frac{zw}{zw+1}$  is the shrinkage factor, and $\widehat \betabf_{LS}$ is the least-squares estimate of $\betabf$.
    \item Repeat steps b-e until convergence. \\
\end{enumerate}

%%% Local Shrinkage Model %%%%%%%%%%%%%%%%%%%%%%%%%%%%%%%%%%%%%
\subsection{MCMC for the Conditional Local Shrinkage Model}
\subsubsection{Full-rank case}
We now develop a novel MCMC algorithm for the local shrinkage model, sampling from the full conditionals using a Metropolis-Hastings sampler. 
First, we take the eigen-decomposition of $\Sigmabf_\X = \X^T\X/n = \V\D\V^T$, where $\V_{p\times p}$ is an orthogonal matrix of eigenvectors, and $\D_{p\times p}$ is a diagonal matrix of eigenvalues. 
$\R$ is transformed such that $\theta = {\R}/({1-\R})$ has a Beta-Prime (or Inverted-Beta distribution), with density $p(\theta) = \theta^{a-1} \left( 1 + \theta \right)^{-a-b} / B(a,b)$.
We also transform $\betabf$ to lie on the unit sphere conditional on the other variables; that is,  $\gammabf = \D^{1/2}\V^T\betabf / \sqrt{\theta\sigma^2}$.
Then $\gammabf | \Lambdabf$ has a Bingham distribution, and we can write the full model as follows.
\begin{flalign*}
    \ybf \mid \gammabf, \sigma^2, \theta &\sim  N_n\left(\sqrt{\theta\sigma^2}\Xbf \V \D^{-1/2} \gammabf ,~\sigma^2 I\right)\\ 
    %\gammabf \mid \Lambdabf& \sim  N_p\left( \zerobf, \left( \D^{-1/2}\V' \Lambdabf^{-1} \V \D^{-1/2}/n \right)^{-1} \right)
    %\ind \left\{ \gammabf'\gammabf=1 \right\}\\
    \gammabf \mid\Lambdabf& \sim \text{ Bingham}\left(\D^{-1/2}\V^T \Lambdabf^{-1} \V \D^{-1/2} \right)\\
    \sigma^2 &\sim \text{ Inverse-Gamma}(a_1, b_1)\\
    \theta &\sim \text{ Beta-Prime}(a,b)\\
    \lambda_j &\sim \text{ Gamma}(\nu, \mu), \text{ for } j=1, \dots, .p
%    \rho &\sim \text{ gamma}(\nu p +1, 1).
\end{flalign*}
%Note that $\gammabf|\Lambdabf$ is the Bingham distribution with matrix parameter $\A=\D^{-1/2}\V' \Lambdabf^{-1} \V \D^{-1/2}/2$, and does not depend on $\sigma^2$ or $\theta$.
%Specifically, $p(\gammabf| \Lambdabf) = C_{X, \Lambda} \, exp\left\{ -\gammabf' \A \gammabf  \right\} $ with respect to the uniform measure on the $p-1$ dimensional unit sphere, and $C_{X, \Lambda}$ is an intractable normalization constant -- the confluent hypergeometric function with matrix argument $\A$.
This parametrization of the prior models the direction of $\betabf$ independently of $\sigma^2$ and $\R$.  
However, the Bingham distribution of the direction, $\gammabf\mid  \Lambdabf$, contains an intractable normalizing constant, $C_{X, \Lambda}$, depending on $\X$ and $\Lambdabf$.
Specifically, $C_{X, \Lambda}$ is a confluent hypergeometric function with matrix argument $\A =\D^{-1/2}\V^T \Lambdabf^{-1} \V \D^{-1/2}/2 $.

The full conditional posterior distribution of $\gammabf$ is a Fisher-Bingham($\mubf, \A$) distribution \citep{kent1982fisher}, where $\mubf = \ybf^T\X \V\D^{-1/2} \sqrt{\theta/\s}$.
This is equivalent to a $N_p(\A^{-1}\mubf, \A^{-1})$ distribution conditioned to lie on the $p-1$ dimensional unit sphere.
We sample from this using the rejection sampler proposed by \cite{kent2013new} with an Angular Central Gaussian (ACG) envelope distribution. 
Sampling efficiently from the ACG distribution is possible because it is just the marginal unit direction of a multivariate Normal distribution with mean $\zerobf$, and thus only requires draws from a Normal distribution.
Any standard MCMC algorithm can sample $\theta$ and $\sigma^2$, but the Adaptive Metropolis algorithm \citep{haario2001adaptive} automatically accounts for the strong negative correlation between the parameters without the need for manual tuning.  
A bivariate Normal proposal distribution is used for ($\sigma^2, \theta$) with covariance proportional to the running covariance of the samples during the burn-in phase.

The density function of the full conditional posteriors for variance parameters, $\lambda_j$, almost have Generalized Inverse Gaussian distributions (GIG) if it were not for the intractable term $C_{X, \Lambda}$. 
A Metropolis-Hastings algorithm would require computing this quantity. 
Our solution is to propose each candidate $\lambda^*_j$ from a GIG distribution, and introduce auxiliary variables, $\gammabf^*|\Lambdabf^*$, from a Bingham distribution in which the constant, $C_{X, \Lambda^*}$, appears in the density. 
We calculate the Metropolis-Hastings acceptance probability for ($\Lambdabf^*, \gammabf^*$), and since $C_{X, \Lambda^*}$ appears in the posterior and the proposal distribution, we avoid computing it.
This is the so-called ``Exchange Algorithm'' proposed by \cite{murray2006} for doubly-intractable distributions, and also used by \cite{fallaize2016exact} for Bayesian inference of the Bingham distribution. 

The entire sampler for the local shrinkage model is described as follows:
\begin{enumerate}[label=(\alph*)]
    \item Set initial values for  $\gammabf, \sigma^2, \theta, $ and $\Lambdabf$.
    \item Sample $\gammabf\mid\Lambdabf$ from a Fisher-Bingham$\left( \mubf, \A \right)$ distribution, 
        where $\mubf = \ybf^T\X \V\D^{-1/2} \sqrt{\theta/\s}$ and $\A=\D^{-1/2}\V^T \Lambdabf^{-1} \V \D^{-1/2}/2$.
    \item Sample ($\sigma^2, \theta $) jointly using an Adaptive Metropolis algorithm from a bivariate Normal distribution. 
    \item Exchange algorithm to sample $\Lambdabf$:
        \begin{enumerate}[label=(\roman*)]
            \item Sample $\lambda^*_j \mid u_j \sim \text{ GIG}(\nu, 2\mu , u_j^2)$, for $j = 1, \dots, p$, where $\ubf = \V\D^{-1/2}\gammabf$.
            %    The GIG distribution has density: $p(x; d, a, b) \propto x^{d-1} exp\left\{ -(ax + b/x)/2\right\}$, \\for $-\infty<d<\infty$, and $ a,b>0$. TODO: move to paragraph
            \item Sample $\ubf^*$ from a Bingham($\A^*$) distribution, where $\A^*=\D^{-1/2}\V^T\Lambdabf^{*-1}\V\D^{-1/2}/2$ 
            \item Accept ($\Lambdabf^*, \ubf^*)$ with probability:
            \begin{align*}
               & \frac{p(\Lambdabf^*|\ubf)}{p(\Lambdabf|\ubf)}\times \frac{q(\Lambdabf|\ubf)}{q(\Lambdabf^*|\ubf)}\times
               \frac{q(\ubf^*|\Lambdabf)}{q(\ubf^*|\Lambdabf^*)}
                = exp\left\{ \sum_{j=1}^p u_j^{*2} \left( 1/\lambda_j^{*2} -  1/\lambda_j^{2}\right)\right\} 
            \end{align*}
        \end{enumerate}
    \item Repeat steps (b)-(d) until convergence, and calculate $\betabf = \sqrt{\theta \s} \V \D^{-1/2} \gammabf$ for each sample. 
\end{enumerate}

In step (d), $p(\Lambdabf\mid\ubf)$ is the conditional posterior distribution of $\Lambdabf$; 
$q(\Lambdabf\mid\ubf)$ is the GIG proposal distribution with density:  $p(x; d, a, b) \propto x^{d-1} exp\left\{ -(ax + b/x)/2\right\}$, \\for $-\infty<d<\infty$, and $ a,b>0$;
and $q(\ubf\mid\Lambdabf)$ is the Bingham proposal distribution which has the same constant as in $p(\Lambdabf\mid\ubf)$.
We only need to keep $\Lambdabf^*$ at each step, and can discard $\ubf^*$.
We can efficiently sample from the Bingham distribution because it is a special case of the Fisher-Bingham with $\mubf = \zerobf$ \citep{kent2013new}. 
All modeling is done in terms of $\gammabf, \s, \theta$ and $\Lambdabf$, and we calculate $\betabf$ outside the sampler. 

%%% High Dimensional Data %%%%%%%%%%%%%%%%%%%%%%%%%%%%%%%%%%%%%
\subsubsection{Non full-rank case}
Next we address how to fit these models with high-dimensional data, where $p>n$ and $\X^T\X$ is not full rank. 
The restriction on $\betabf: \betabf^T \Sigmabf_\X\betabf= \s\theta$, is no longer an ellipsoid, but an unbounded subspace in $p$-dimensions (e.g. parallel lines for $p=2$, and an infinite cylinder for $p= 3$). 
We assume that the rank$(\X)=n$, and partition $\D = \left( \begin{smallmatrix} \D_1 & \zerobf \\ \zerobf & \zerobf \end{smallmatrix} \right)$ and $\V = (\V_1, \V_2)$. 
Note that $\D_1$ is the $n\times n$ diagonal matrix of positive eigenvalues, $\V_1$ is the matrix of corresponding eigenvectors, and $\V_2$ is the matrix of the $p-n$ eigenvectors spanning the null space of $\X$. 
We define $\gammabf = (\gammabf_1, \gammabf_2)^T = (\D_1^{1/2}\V_1^T\betabf, \V_2^T\betabf)^T/\sqrt{\theta \s}$, so that $\gammabf$ is multivariate Normal with the constraint that $\gammabf_1^T\gammabf_1=1$.
Marginally, $\gammabf_1 = \D_1^{1/2}\V_1^T\betabf/\sqrt{\theta \s}$ is defined on the $n-1$ dimensional unit sphere, and has a Fisher-Bingham distribution just like the full rank case.  
However, the reverse transformation $\betabf =  \sqrt{\theta \s} \V_1 \D_1^{-1/2} \gammabf_1$ is defined on the lower $n-1$ dimensional ellipsoid within the entire constrained space $\{\betabf: \betabf^T\Sigmabf_\X\betabf= \s\theta\}$.
For example, if $p = 3$ and $n = 2$, this is the slice of the 3-dimensional ellipsoid with the minimum $L_2$-norm.
The problem is that this lower dimensional ellipsoid may not be able to favor the sparsity or local shrinkage encouraged by $\Lambdabf$. 
That would be equivalent to principal components regression using the top $n$ principal components. 
To allow for shrinkage of the original coefficients and not the principal components, we sample $\gammabf_2\mid\gammabf_1$, which is multivariate Normal, and make the reverse transformation $ \betabf = \sqrt{\theta \s}\left(\V_1\D_1^{-1/2}\gammabf_1 + \V_2\gammabf_2\right) $.
Since $\V_2$ spans the null space of $\X$, $\betabf$ is still in the constrained region, but offers more flexibility in shrinking $\beta_j$. 

%-------------------------------------------------------------------------------  
\subsection{Gibbs Sampler for Marginal R2-D2}
%------------------------------------------------------------------------------- 
For posterior computation in the marginal case, the following equivalent representation is useful. 
The R2-D2 prior  (\ref{eq_R2-D2 prior_dirichlet BP form}) is  equivalent to 
\begin{eqnarray} \label{eq_R2-D2 prior, linear model, normal representation}
\beta_j\mid\sigma^2, \psi_j, \phi_j, \omega   \sim   \text{N}(0,\psi_j\phi_j \omega\sigma^2/2),  \ 
\psi_j  \sim    \text{Exp}(1/2), \  \nonumber  \\
{\phi}  \sim \text{Dir}(a_\pi,\cdots,a_\pi),    \
\omega  \mid \xi     \sim  \text{Ga}(a,\xi),  \ 
\xi  \sim   \text{Ga}(b,1), 
\end{eqnarray} 
where $\text{Exp}(\delta)$ denotes the exponential distribution with mean $\delta^{-1}$. 
The Gibbs sampling procedure  is    based on (\ref{eq_R2-D2 prior, linear model, normal representation}) with $a = pa_\pi$.  	Denote $Z\sim \text{InvGaussian}(\mu, \lambda)$, 
%the inverse Gaussian distribution, 
if $ \pi(z) =  {\lambda^{1/2}(2\pi z^3)}^{-1/2}\text{exp}\{-  {\lambda(z-\mu)^2}/(2\mu^2z)\}$. 
Denote $Z\sim\text{giG}(\chi,\rho,\lambda_0)$, the generalized inverse Gaussian distribution \citep{seshadri1997halphen}, if $\pi(z) \propto z^{\lambda_0-1}\text{exp}\{- (\rho z + \chi/z)/2\}$. 

The Gibbs sampling procedure is as follows: 
\begin{enumerate}[label=(\alph*)]
	%--------%--------%--------%--------%--------%--------%--------
	% Sample beta
	%--------%--------%--------%--------%--------%--------%--------
	 \item Set initial values for  $\betabf, \sigma^2, \bm{\psi},\bm{\phi},$ and $\omega$.
	 \item Sample 
	$\bm\beta \mid \bm{\psi},\bm{\phi} ,\omega,   \sigma^2, {\bm Y} \sim N(\bm\mu,\sigma^2 {\bm V})$, where $ \bm\mu ={\bm V} \bm{X}^T\bm{Y} = (\bm{X}^T\bm{X}+{\bm S}^{-1})^{-1} (\bm{X}^T\bm{Y})$,   $ {\bm V} = ( \bm{X}^T\bm{X} +{\bm S}^{-1})^{-1} $,   ${\bm S} = \text{diag}\{ \psi_1\phi_1\omega/2, \cdots, \psi_p\phi_p\omega/2 \}$,  $\bm{X} = (\bm{x}_1,\cdots, \bm{x}_n)^T$, and $\bm{Y} = (Y_1, \cdots, Y_n)^T$. 
	
	%--------%--------%--------%--------%--------%--------%--------
	% Sample \sigma^2
	%--------%--------%--------%--------%--------%--------%--------
	
	\item  Sample $\sigma^2 \mid  \bm\beta,\bm{\psi}, \bm{\phi},\omega,{\bm Y}  \sim \text{IG}(a_1+  (n+p)/{2}, b_1 + ( \bm\beta^T {\bm S}^{-1}\bm\beta  +(\bm{Y}-\bm{X}\bm\beta)^T  (\bm{Y}-\bm{X}\bm\beta) )/{2})$.

	%--------%--------%--------%--------%--------%--------%--------
	% Sample psi
	%--------%--------%--------%--------%--------%--------%--------
	\item Sample $\bm{\psi}\mid \bm\beta, \bm{\phi},\omega, \sigma^2$. Draw $ {\psi_j}^{-1}\sim \text{InvGaussian}(\mu_j =   {\sqrt{\sigma^2\phi_j \omega/2}}/{|\beta_j| }, \lambda=1 $), then take the reciprocal to get $\psi_j$. 
	
	%--------%--------%--------%--------%--------%--------%--------
	% Sample \omega
	%--------%--------%--------%--------%--------%--------%--------
	\item Sample $\omega\mid \bm\beta, \bm{\psi},\bm{\phi},  \xi, \sigma^2 \sim \text{giG}(\chi = \sum\limits_{j=1}^p {2\beta_j^2} /(\sigma^2\psi_j\phi_j), \rho = 2\xi, \lambda_0 = a-p/2)$. 
	%--------%--------%--------%--------%--------%--------%--------
	% Sample xi
	%--------%--------%--------%--------%--------%--------%--------
	\item Sample $\xi \mid \omega \sim \text{Ga}(a+b, 1+ \omega)$. 
	
	%--------%--------%--------%--------%--------%--------%--------
	% Sample phi 
	%--------%--------%--------%--------%--------%--------%--------
	\item Sample $\bm\phi\mid \bm\beta,\bm\psi,  \xi, \sigma^2$.  
	Motivated by \cite{bhattacharya2015dirichlet}, if  $a = pa_\pi$, one can draw $T_1, \cdots, T_p$ independently with  $T_j\sim \text{giG} (\chi = {2\beta_j^2}/(\sigma^2\psi_j), \rho = 2\xi, \lambda_0 = a_\pi-1/2)$. Then set $\phi_j = T_j/T$ with $T=\sum\limits_{j=1}^p T_j$.  
	 \item Repeat steps (b)-(g) until convergence.  
\end{enumerate}

%%% Hyperparameter selection %%%%%%%%%%%%%%%%%%%%%%%%%%%%%%%%%%%%%%%%%%%%%%%%%%

%%%%%%%%%%%%%%%%%%%%%%%%%%%%%%%%%%%%%%%%%%%%%%%%%%%%%%%%%%%%%%%%%%%%%%%%%%%%%%%%
%%% Simulations
%%%%%%%%%%%%%%%%%%%%%%%%%%%%%%%%%%%%%%%%%%%%%%%%%%%%%%%%%%%%%%%%%%%%%%%%%%%%%%%%
\section{Simulation Study}\label{sim.section}
We conduct a simulation study to compare the proposed approach with other Bayesian regression models. In each setting, 200 datasets are simulated from the homoscedastic linear model  (\ref{eq_linear_regression}), with sample size  $n = 60$, and the number of predictors $p$  varying  in $p\in\{50, 100, 500, 2000\}$. Larger sample sizes were also investigated with the comparisons remaining similar. The covariates $ \bm{x}_i$,  $i=1,\cdots,n$,  are generated from multivariate normal distribution with mean zero, and correlation matrix of autoregressive  structure  AR(1)  with correlation  $\rho=$0$.$5 or 0$.$9.  For the regression coefficients $ \bm{\beta}$, we set 
$ \bm\beta= ({\bm 0}_{10}^T,{\bm B_1}^T,{\bm 0}_{30}^T,{\bm B_2}^T,{\bm 0}^T_{p-50})^T$
with  ${\bm 0}_k$ representing the zero vector of length $k$, and  ${\bm B_1}$ and $ \bm B_2$  each of length $5$ nonzero elements. 
The  fractions of  true coefficients with exactly zero values 
are 80\%, 90\%, 98\%, and 99.5\% for $p\in\{50, 100,   500, 2000\} $, respectively,   and the remaining  20\%, 10\%, 2\%, and 0.5\% nonzero elements  ${\bm B_1}$ and ${\bm B_2}$   were  independently generated from  two scenarios: 
 (i) a   Student  $t$ distribution with 3 degrees of freedom to  give heavy tails; (ii) a Uniform(0,1)  distribution to give weaker signals. For scenario (i), we set the error variance to produce a Signal-to-Noise Ratio (SNR) of 9, yielding $\sigma^2 = 10/3$. Note that since the generated coefficients have expectation of zero, the SNR does not depend on the correlation structure of the covariates. For scenario (ii), the nonzero expectation of the Uniform(0,1) then entails a different SNR for the 2 different correlation setups. For this scenario, we set $\sigma^2 = 6$ to yield SNR $\approx$ 0.5 and 0.7 in the $\rho = 0.5$ and $0.9$ cases, respectively.    

To implement both the marginal and conditional R2-D2 priors,  as discussed as a default choice in Section 4.2.4, we set $b = 0.5$ to yield Cauchy-like tails, and
$a_\pi = C/(p_n^{b/2} n^{rb/2} \log n)$ with choosing the arbitrary constants $C$ and $r$ to be 1.  We also set $b=0.1$ to give heavier tails (then setting $a_\pi$  in the same manner). The results for $b = 0.1$ were similar to $b = 0.5$ and thus not shown, and we recommend as a default to use $b = 0.5$ as a fully automatic approach. 
 For the conditional R2-D2, the choices of $\nu$ and $\mu$ in the Gamma prior are based on the recommendations for implementation of Normal-Gamma priors given in \cite{griffin2010inference}. The choice is based on the degree of sparsity, which we choose as $\min\{n, 0.1p\}$ so that we set the expected number of non-zero coefficients to be $10\%$ of the total, or the sample size, whichever is smaller. Details of how $\nu$ and $\mu$ relate to the sparsity level are given in the Appendix. 

The comparisons are made to some current state-of-the-art global-local priors: Horseshoe,  Horseshoe+, Normal-Beta Prime and Dirichlet-Laplace.  

For the Horseshoe, Dirichlet-Laplace, and Normal-Beta Prime, the implementation is done via the \verb|R| packages \verb|horseshoe|, \verb|dlbayes|, and \verb|NormalBetaPrime|, respectively. The Horseshoe+ is implemented through \verb|Stan| in \verb|R| using the code provided by the author of \cite{bhadra2016horseshoe+}. The proposed R2-D2 approaches are implemented in \verb|R|, based on the discussed MCMC sampling.  
 
 In all cases, $10,000$ samples are collected with the first $5,000$ samples discarded as burn-in.

{\bf Estimation error and AUC.} The  average sum  of squared error corresponding to the posterior mean   across the 200 replicates is provided in Table  \ref{table-sim-t3} for simulation setting i for $p = $100 and 500 and in Table \ref{table-sim-Uniform} for simulation setting ii.  In addition, the averaged area under the  Receiver-Operating Characteristic (ROC) curve (AUC) based on the posterior $t$-statistic, i.e., the ratio of the posterior mean and  posterior standard deviation, is also given   to offer further evaluation of the reliability of the posterior inference on the coefficients. Larger AUC signifies that the method tends to give posterior intervals further away from zero for the true non-zero coefficients and intervals that are concentrated closer to zero for the irrelevant variables.  
To perform variable selection, thresholding  the posteriors either marginally or jointly \citep{bondell2012consistent} is typical, and thus higher AUC would represent more accurate variable selection. 

\begin{table} [h!]
	\spacingset{1.1}
	\renewcommand{\arraystretch}{1.5}\textbf{}
	\resizebox{\columnwidth}{!}{
		\begin{tabular}{|l|cc|cc||cc|cc|}
			\hline
			\multirow{4}{*}{} &\multicolumn{8}{c|}{Scenario 1 [non-zero coefficients from $t_3$]} \\ \cline{2-9}
			& \multicolumn{4}{c||}{$\rho=0.5$}			
			& \multicolumn{4}{c|}{$\rho=0.9$} \\ \cline{2-9}
			& \multicolumn{2}{c|}{p=100} &
			\multicolumn{2}{c||}{p=500}  & \multicolumn{2}{c|}{p=100} &
			\multicolumn{2}{c|}{p=500} \\ \cline{2-9}
			& SSE & AUC & SSE & AUC & SSE & AUC & SSE & AUC\\
			\hline
			Horseshoe & 148 (6.4) & 64 & 236 (13.1) & 61 & 186 (9.2) & 67 & 205 (11.2) & 68 \\
			\hline
			Horseshoe+  & 146 (5.8) & 64  & 215 (11.4) & 62 & 190 (9.7) & 67  & 212 (11.4) & 69 \\
			\hline
			Normal-BetaPrime  & 155 (5.7) & 64  & 578 (20.1) & 59 & 183 (8.1) & 67  & 1399 (62.1) & 57 \\
			\hline
			Dirichlet-Laplace  & 166 (7.8) & \textbf{67}  & 187 (8.1) & 54 & 197 (11.6) & \textbf{74} & 209 (10.9) & 59 \\
			\hline
			R2-D2 - Conditional & 146 (5.6) & 66 & 179 (7.9) & 62 & \textbf{174 (8.1)} & 73 & 192 (11.2) & 72 \\
			\hline	
			R2-D2 - Marginal & \textbf{142 (5.5)} & 65 & \textbf{166 (6.5)} & \textbf{64} & 179 (8.7) & 72 & \textbf{188 (10.1)} & \textbf{73}  \\
			\hline
	\end{tabular}}
	\caption {\small Average sum of squared error (SSE) and average  area under the  Receiver-Operating Characteristic curve (AUC), based on 200 simulated datasets. All values  multiplied by 10 for readability.  For SSE, standard errors are included in parentheses. For AUC, all standard errors are in the range of $0.5 - 0.7$, so omitted to save space. Note that {\it smaller} is better for SSE, while {\it larger} is better for AUC. Best performance in each column is highlighted in bold for reference.} 
	\label{table-sim-t3}
\end{table}

\begin{table} [h!]
	\spacingset{1.1}
	\renewcommand{\arraystretch}{1.5}\textbf{}
	\resizebox{\columnwidth}{!}{
		\begin{tabular}{|l|cc|cc||cc|cc|}
			\hline
			\multirow{4}{*}{} &\multicolumn{8}{c|}{Scenario 2 [non-zero coefficients from U(0,1)]} \\ \cline{2-9}
			& \multicolumn{4}{c||}{$\rho=0.5$}			
			& \multicolumn{4}{c|}{$\rho=0.9$} \\ \cline{2-9}
			& \multicolumn{2}{c|}{p=100} &
			\multicolumn{2}{c||}{p=500}  & \multicolumn{2}{c|}{p=100} &
			\multicolumn{2}{c|}{p=500} \\ \cline{2-9}
			& SSE & AUC & SSE & AUC & SSE & AUC & SSE & AUC\\
			\hline
			Horseshoe  & 44.1 (1.9) & 62 & 59.9 (4.8) & 58 & 41.7 (2.2) & 73 & 39.8 (3.2) & 58 \\
			\hline
			Horseshoe+ & 47.0 (2.2) & 62  & 79.2 (8.1) & 61 & 47.6 (4.2) & 73  & 44.9 (4.4) & 63 \\
			\hline
			Normal-BetaPrime  & 62.4 (2.4) & 62  & 437.1 (14.2) & 57 & 42.5 (1.7) & 70  & 1203 (55.1) & 59 \\
			\hline
			Dirichlet-Laplace  & \textbf{32.4 (0.7)} & 64  & 56.3 (5.7) & 51 & 27.3 (0.6) & 79 & 41.3 (1.2) & 64 \\
			\hline
			R2-D2 - Conditional & 34.5 (1.2) & \textbf{66} & 48.2 (2.9) & 59 & \textbf{26.3 (0.6)} & \textbf{82} & 38.1 (1.7) & 79 \\
			\hline	
			R2-D2 - Marginal & 33.7 (1.0) & 64 & \textbf{41.0 (1.3)} & \textbf{65} & 31.8 (1.2) & 76 & \textbf{34.5 (1.0)} & \textbf{82} \\
			\hline
	\end{tabular}}
	\caption {\small Average sum of squared error (SSE) and average  area under the  Receiver-Operating Characteristic curve (AUC), based on 200 simulated datasets. All values  multiplied by 10 for readability.  For SSE, standard errors are included in parentheses. For AUC, all standard errors are in the range of $0.5 - 0.7$, so omitted to save space. Note that {\it smaller} is better for SSE, while {\it larger} is better for AUC. Best performance in each column is highlighted in bold for reference.} 
	\label{table-sim-Uniform}
\end{table}	

It is clear from Table \ref{table-sim-t3} and \ref{table-sim-Uniform} that the proposed approaches, both conditional and marginal R2-D2 prior, outperform the existing methods in nearly every case, both in terms of the estimation error, and the area under the ROC curve. This is particularly apparent in the cases of $p=500$.
The Dirichlet-Laplace performs well in some of the cases with $p=100$, but then performs significantly worse with $p=500$, particularly in AUC. 
%Meanwhile, while the Horseshoe and Horseshoe+ are closer to R2-D2 in terms of their estimation error, they have significantly worse area under the ROC curve.
 Both the conditional and marginal versions of R2-D2 prior perform similarly, with slightly better performance of the conditional version  in the lower dimensional case, but slightly better performance of the marginal version in the case of $p = 500$. Overall, the proposed approach gives a much improved result over the existing approaches. This is to be anticipated based on the theoretical results on the concentration and tail behaviors, which we now examine.

% We also plot the averaged  ROC curve over the 200 simulated data sets with $\rho=0.9$ and $p=500$ for both scenarios in Figure \ref{figure-roc-normalbeta prime}.  It shows that for any choice of cutoff,  our proposed marginal R2-D2 prior will get more true positives and less false positives (since our curve dominates the others).  

%\begin{figure}[htbp]
%	    \spacingset{1.1}% DON'T change the spacing!
%	\centering
%	\includegraphics[width=0.49\linewidth]{ROC_Plot_Scenario1_rho09_p500}  
%	\includegraphics[width=0.49\linewidth]{ROC_Plot_Scenario2_rho09_p500}
%	\caption{ROC curves averaged over the 200 simulated datasets with $\rho=0.9$ and $p=500$. Left panel: Scenario 1 (stronger signal). Right Panel: Scenario 2 (weaker signal).}
%		\label{figure-roc-normalbeta prime}
%\end{figure}

To better understand the relative performance of the methods, Table \ref{table-breakdown} shows the average SSE partitioned according to the value of the true $ \bm{\beta}$ at  $\beta_j=0$, $|\beta_j|\in (0,0.5]$,  and $|\beta_j| > 0.5$,  $j=1,\cdots,p$. This allows us to view a more detailed performance of the approaches in their behavior on the zero, small non-zero, and larger non-zero coefficients, respectively. The table shows the results for $\rho = 0.5$ and $p = 100$ for setting i, but the decomposed SSE is similar in the other cases as well.

\begin{table}[htbp]
	\spacingset{1.1}% DON'T change the spacing!
	\centering
	\renewcommand{\arraystretch}{1.2}
	\begin{tabular}{|l|c|c|c|c|}
		\hline
		& $\beta = 0 $ & $|\beta| \in(0,0.5]$ & $|\beta| > 0.5$ & Total \\ 
		\hline
		Horseshoe & 21 & 24 & 103 & 148\\
		\hline 
		Horseshoe+& 18 & 22 & 106 & 146 \\      
		\hline
		Normal-BetaPrime & 35 & 38 & 82 & 155 \\
		\hline
		Dirichlet-Laplace & 4 & 6 & 156 & 166 \\ 
		\hline
		R2D2 - Conditional & 14 & 18 & 114 & 146 \\  
		\hline
		R2D2 - Marginal & 10 & 13 & 119 & 142 \\   
		\hline
	\end{tabular}
	\caption{\normalsize Average sum of squared error broken down by zero coefficients $(\beta = 0 )$, small coefficients $(|\beta| \in(0,$0$.$5$]$), large coefficients $(|\beta| > $0$.$5$]$), and Total. Results are based on the 200 datasets from Table \ref{table-sim-t3} with $p = 100$ and $\rho = 0.5$.}
	\label{table-breakdown}
\end{table}

From the breakdown provided in Table  \ref{table-breakdown}, we see the differences between the approaches, and the theoretical results on the concentration at zero and tail behavior (Table \ref{table-asymptotics}) really show up strongly here. Since the Dirichlet-Laplace has high concentration around zero, but it has tails that are much lighter than the others, we see this translate into really small error on the zero and small coefficients, but extremely large error on the large coefficients. Meanwhile, the Horseshoe and Horseshoe+ having less concentration around zero leads to poorer performance at estimation of the actual zeros. Note that this also explains the reason for the Horseshoe and Horseshoe+ having worse AUC, as it does not push the zeros close enough to zero in order to distinguish them from the small coefficients. Note that the R2-D2 was set with $b = 0.5$ giving it the Cauchy-like tails as is the case with the Horseshoe. Thus we see similar estimation ability for the larger coefficients. Overall, the proposed R2-D2 approach achieves a strong performance in both regions simultaneously as anticipated by the theory.

{\bf Credible Interval Coverage.} We  also examine the coverage properties of using 95\% marginal credible intervals for each approach. For this same scenario, table \ref{table-cover} reports the average width of the intervals, the proportion of coverage, the Specificity, and the Sensitivity. This gives a broader picture of the posterior inference properties. Although, there is no promise of 95\% Frequentist coverage by the 95\% intervals, we see that in this case of $n = 60$ and $p = 100$, the coverage of most of the approaches are close to 95\%, with the proposed approaches and the Horseshoe+ being almost right on the value, while the Horseshoe and DL are a bit lower than the target, and the NBP being a bit large. For the $n=60$ case, this is a very difficult problem even when $p=100$, and we see that the sensitivity  (power) for all approaches are quite low due to the intervals containing zero even for the majority of the true non-zeros. The sensitivity for the R2-D2 approaches are highest, along with the Horseshoe+ and Normal-Beta Prime. However, in order to do so, the Horseshoe+ and Normal-Beta Prime intervals are quite wide in comparison, thus showing that our proposed approaches look to have good posterior inference properties in addition to their outstanding performance in estimation and variable importance ordering.

\begin{table}[htbp]
	\spacingset{1.1}% DON'T change the spacing!
	\centering
	\renewcommand{\arraystretch}{1.2}
	\begin{tabular}{|l|c|c|c|c|}
		\hline
		& Avg Width & Coverage & Sensitivity & Specificity \\ 
		\hline
		Horseshoe & 0.81 & 0.936 & 0.190 & 0.999 \\
		\hline 
		Horseshoe+& 1.23 & 0.953 & 0.310 & 0.999 \\      
		\hline
		Normal-BetaPrime & 1.83 & 0.972 & 0.301 & 0.999 \\
		\hline
		Dirichlet-Laplace & 0.91 & 0.939 & 0.150 & 1 \\ 
		\hline
		R2D2 - Conditional & 0.98 & 0.947 & 0.350 & 0.999 \\  
		\hline
		R2D2 - Marginal & 1.11 & 0.948 & 0.290 & 1 \\   
		\hline
	\end{tabular}
	\caption{\normalsize Coverage properties of 95\% marginal posterior credible intervals. Results are based on the 200 datasets from Table \ref{table-sim-t3} with $p = 100$ and $\rho = 0.5$.}
	\label{table-cover}
\end{table}

{\bf Higher dimensional setting.} Table  \ref{table-largep} shows the results for $p = 2000$ and $n=60$ for scenario i with $\rho = 0.5$. Other settings are similar and hence not shown. Table  \ref{table-largep} shows the estimation error and AUC as before. We also include coverage of the 95\% intervals, and also the coverage on only the non-zero coefficients. For estimation and AUC, we again see that the performance of the proposed R2-D2 approaches perform very well in this higher dimensional case.
 
In terms of coverage of the 95\% intervals, as shown in \cite{van2017uncertainty} (Theorems 1 and 2), for the HS credible intervals in the special case of the Normal means model, the coverage will either go to zero or one depending on the size of the true coefficient. We see that in this case of $p = 2000$ and $n = 60$, all methods have overall coverage of nearly 100\%. However, stark differences appear in looking at the coverage on the non-zero coefficients only (of which there are only 10 out of 2000 here). We see that the HS, in particular, almost never covers these non-zero values, with coverage of 1.5\% (HS) or 5\% (HS+). This is anticipated by the results of \cite{van2017uncertainty} given that it attempts to adapt to the sparsity, thus covering the zeros, at the expense of the non-zeros. We see that the other methods do better at coverage of the non-zeros here (with the marginal R2-D2 way up at 47\%). We see in this higher dimensional case, the proposed approaches exhibiting more stability, both in estimation error, and in the coverage.

\begin{table}[htbp]
	\spacingset{1.1}% DON'T change the spacing!
	\centering
	\renewcommand{\arraystretch}{1.2}
	\begin{tabular}{|l|c|c|c|c|}
		\hline
		& \multicolumn{4}{c|}{$p = 2000$, $n = 60$ [non-zero coefficients from $t_3$]} \\
		\hline
		& SSE & AUC & Coverage & Coverage on non-zeros\\ 
		\hline
		Horseshoe & 249 (9.6) & 44 & 0.995 & 0.015 \\
		\hline 
		Horseshoe+& 233 (8.9) & 51 & 0.996 & 0.050 \\      
		\hline
		Normal-BetaPrime & 468 (15.7) & 61 & 0.995 & 0.240 \\
		\hline
		Dirichlet-Laplace & 240 (8.7) & 56 & 0.996 & 0.165 \\ 
		\hline
		R2D2 - Conditional & 201 (7.1) & \textbf{65} & 0.996 & 0.271 \\  
		\hline
		R2D2 - Marginal & \textbf{196 (6.9)} & 64 & 0.997 & \textbf{0.466} \\   
		\hline
	\end{tabular}
	\caption{\normalsize Simulation setting 1 with $p = 2000$, $n = 60$, and $\rho = 0.5$. Average sum of squared error (with standard error), area under the ROC curve, overall coverage of 95\% credible intervals on all 2000 coefficients, and coverage of the 95\% credible intervals on the 10 non-zero coefficients only. Results are based on 200 datasets.}
	\label{table-largep}
\end{table}

%%%%%%%%%%%%%%%%%%%%%%%%%%%%%%%%%%%%%%%%%%%%%%%%%%%%%%%%%%%%%%%%%%%%%%%%%%%%%%%%
%%% Real data examples
%%%%%%%%%%%%%%%%%%%%%%%%%%%%%%%%%%%%%%%%%%%%%%%%%%%%%%%%%%%%%%%%%%%%%%%%%%%%%%%%
\section{Data Examples}
We study the predictive performance of the posterior generated by the  R2-D2  prior through a variety of real examples which exhibit varying structures. These three data sets have many more parameters than observations, and have very different correlation structures. 
The Cereal data consists of starch content measurements from 15 observations with 145 infrared spectra measurements as predictors. 
The data is provided with the \verb#chemometrics# R package. 
The Cookie data arises from an experiment testing the near-infrared (NIR) spectroscopy of biscuit dough in which the fat content is measured on 72 samples, with 700 NIR spectra measurements as predictors.
The data was generated in the experiment by \cite{osborne1984application}, and is available in the \verb#ppls# R package. 
The Multidrug data are from a pharmacogenomic study investigating the relationship between the drug concentration (at which 50\% growth is inhibited for a human cell line) and expression of the adenosine triphosphate binding cassette transporter \citep{szakacs2004predicting}. 
The data consists of 853 drugs as predictors, 60 samples of human cell lines using the ABCA3 transporter as the response, and is available in the \verb#mixOmics# R package. 
In the statistics literature, the Cereal and Multidrug data were both studied by \cite{polson2012local} and \cite{griffin2013some}; and the Cookie data was studied by \cite{brown2001bayesian} and \cite{ghosh2015bayesian}. 

The three datasets nicely represent 3 different correlation structures among the predictor variables. The Multidrug covariates have low to moderate pairwise correlations, the Cookie covariates are highly positively correlated, and the Cereal covariates have a wide range that are both positively and negatively correlated. 
Figure \ref{figure-histogram} shows histograms of all pairwise correlations for each of the data sets.

\begin{figure}[htbp]
	    \spacingset{1.1}% DON'T change the spacing!
	\centering
	\includegraphics[width=0.3\linewidth]{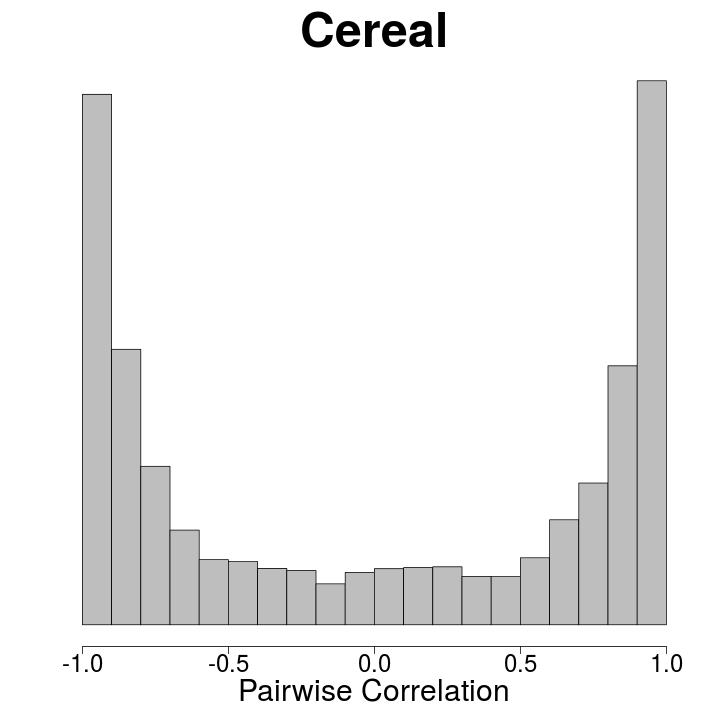}  
	\includegraphics[width=0.3\linewidth]{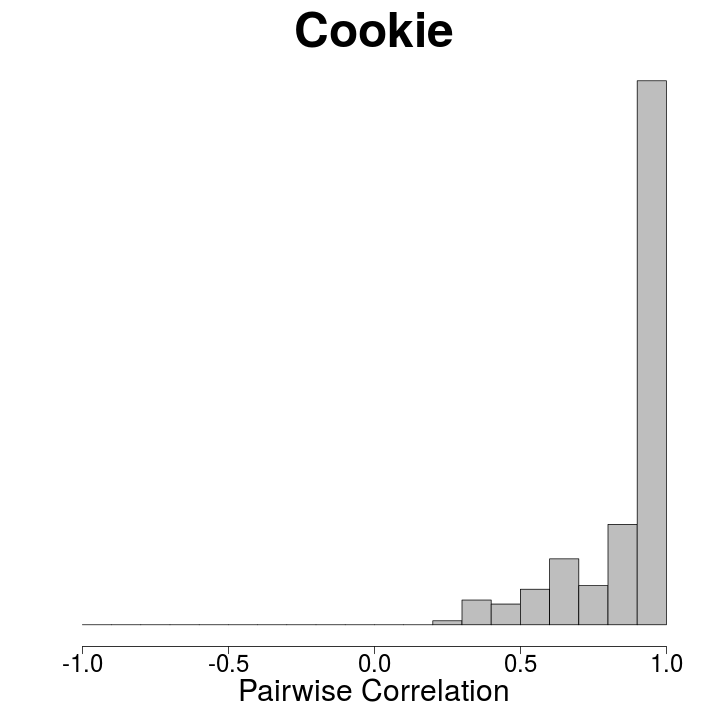}  
	\includegraphics[width=0.3\linewidth]{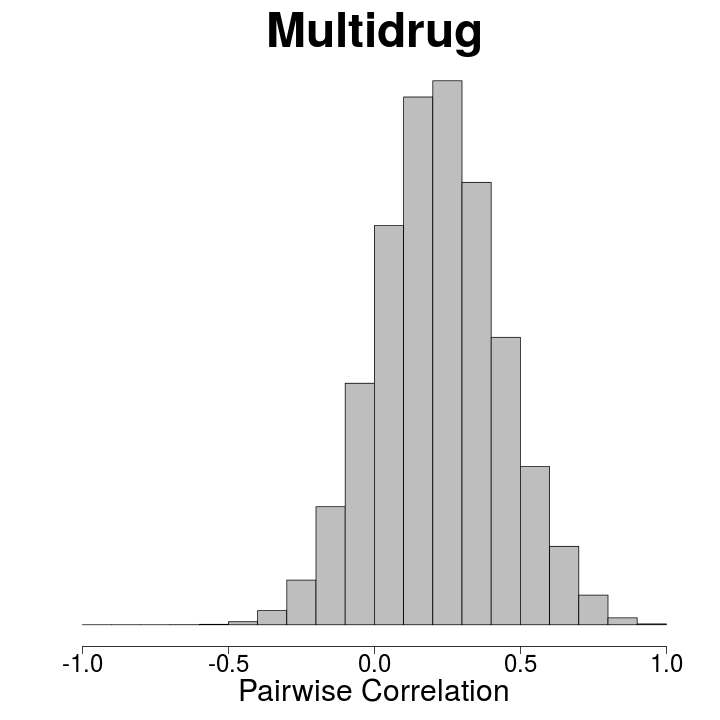}  
	\caption{Histograms of pairwise correlations among the predictor variables for the three datasets.}
	\label{figure-histogram}
\end{figure}

\begin{table}[htbp]
	\centering
	\spacingset{1.1} % DON'T change the spacing!
		\renewcommand{\arraystretch}{1.2}
	\begin{tabular}{|l|c|c|c|}
		\hline
		& Cereal & Cookie & Multidrug \\ 
		\hline
		$n$ & 15 & 72 & 60 \\ 
		$p$ & 145 & 700 & 853 \\ 
		\hline
		Horseshoe  & 14.1 (1.0) & 8.3 (0.2)& 15.6 (0.5) \\ 
		\hline
		Horseshoe+ & 14.2 (0.9) & 9.1 (0.3)  & 15.0 (0.5) \\ 
		\hline
		Normal-BetaPrime & 25.4 (1.5) & 11.9 (0.4)  & 18.4 (0.6) \\ 
		\hline
		Dirichlet-Laplace & 15.1 (1.2) & 12.1 (0.5)  & \textbf{12.2 (0.3)} \\ 
		\hline
		R2-D2 - Conditional & 12.2 (0.5) & 9.8 (0.3)  & 12.7 (0.3) \\ 
		\hline
		R2-D2 - Marginal & \textbf{12.1 (0.5)} & \textbf{8.1 (0.2)} & 12.6 (0.3) \\ 
		\hline
	\end{tabular}
	\caption{Average mean square prediction error (and standard errors) for each of the data examples.}
	\label{table-real}
\end{table}

We randomly split each data set into a training and testing sets to evaluate the out-of-sample predictive performance. 
For each data set, 75\% of the observations were used for training, and the remaining 25\% were used for estimating the mean squared prediction error (MSPE) between the test sample and predictions.
This process was repeated to create 200 data sets for each example. 
The same 5 approaches as in the simulations were used on each of the datasets. Due to the various correlation structures and dimensions, this gives a range of potential data structures for comparison of the approaches. 

The average MSPE results are given in Table \ref{table-real}. We see that the R2-D2 approaches consistently outperform the existing methods across the datasets. Note that the Horseshoe and Horseshoe+ perform well on the Cookie data, but are significantly worse than the others on the other 2 datasets.

%%%%%%%%%%%%%%%%%%%%%%%%%%%%%%%%%%%%%%%%%%%%%%%%%%%%%%%%%%%%%%%%%%%%%%%%%%%%%%%
%% Discussion
%%%%%%%%%%%%%%%%%%%%%%%%%%%%%%%%%%%%%%%%%%%%%%%%%%%%%%%%%%%%%%%%%%%%%%%%%%%%%%%
\section{Discussion}

In this paper, we propose a  shrinkage prior motivated by assuming a prior on $R^2$. The prior exhibits polynomial behavior both around the origin and in the tails and compares favorably with other global-local shrinkage priors. Although the motivation of our R2-D2 prior is via starting with a prior on $R^2$, the resultant prior is simply a member of the class of global-local shrinkage priors, which can then be applied directly to other models, as with other priors. The prior is represented by a hierarchical scale mixture of normals, which can then be implemented in a generalized linear model or other regression setting. The hyperparameters would no longer have the interpretation as parameters of a Beta prior on the $R^2$ of the model. But the form and properties of the resulting prior, such as tail behavior and concentration around zero remain directly useable as is with the other global-local priors.  

There is scope for further advances in algorithms for MCMC sampling for these posteriors, just as there has been for sampling from other prior proposals, as for example, key sampling approaches for the high-dimensional case as in \cite{Bhatt2016sampling}.

%\pagebreak
%\setstretch{1}
% \bibliographystyle{ECA_jasa} % Style template
% \bibliographystyle{plainnat} % Style template
\bibliography{LitSummary}
\bibliographystyle{agsm}
% \bibliography{Bibliography-MM-MC}

\clearpage
\appendix
\section{Appendix: Technical Details} 

\textbf{Definition of the Meijer G-function.}

A general definition of the Meijer G-function is given by the following line integral in the complex plane \citep{bateman1953higher}:
\[
G^{m,n}_{p,q}\left(z \left\vert^{a_1, \dots, a_p}_{b_1,\dots, b_q} \right. \right ) = \frac{1}{2\pi i } \int_L \frac{\prod_{j=1}^m \Gamma(b_j-s)\prod_{j=1}^n \Gamma(1-a_j+s)}{\prod_{j=m+1}^q \Gamma(1-b_j+s)\prod_{j=n+1}^p \Gamma(a_j-s) } z^s \, ds 
\]
where $\Gamma(\cdot)$ denotes the gamma function and  $L$  in the integral represents the path to be followed while integrating. The definition holds under the following assumptions: 
\begin{itemize}
	\item $0\leq m\leq q$ and $0\leq n \leq p$, where $m,n,p$ and $q$ are integer numbers
	\item $a_k-b_j\neq 1,2,3,\dots $ for $k=1,2,\dots,n$ and $j=1,2\dots,m$
	\item $z\neq 0$. 
\end{itemize}

\begin{proof} [\textbf{Proof of Proposition \ref{prop1}}]
Derivation of Equation (\ref{eq:priorbetabf}):
Let $\gammabf$ be uniformly distributed on the $p-1$ dimensional unit sphere. 
That is, 
\begin{align*}
p(\gammabf) &= \frac{\Gamma(p/2)}{2\pi^{p/2}} \onebf\left\{ \gammabf^T\gammabf = 1 \right\}. 
\end{align*} 
Define $\theta = \frac{\R}{1-\R} \sim $ BP($a,b$).
Make the transformation $r = \sqrt{\theta}$:
\begin{align*}
p(\gammabf,\theta) &= p(\gammabf)p(\theta) = \frac{\Gamma(p/2)}{2\pi^{p/2}B(a,b)} \theta^{a-1} \left( 1+\theta \right)^{-a-b} \onebf\left\{ \gammabf^T\gammabf = 1 \right\}, \\ \text{and }
p(\gammabf,r) &= \frac{\Gamma(p/2)}{2\pi^{p/2}B(a,b)} r^{2a-2} \left( 1+r^2 \right)^{-a-b} \onebf\left\{ \gammabf^T\gammabf = 1 \right\} |2r|\\
&= \frac{\Gamma(p/2)}{\pi^{p/2}B(a,b)} r^{2a-1} \left( 1+r^2 \right)^{-a-b} \onebf\left\{ \gammabf^T\gammabf = 1 \right\}. 
\end{align*} 
Make the transformation $\zbf = r\gammabf$. 
The Jacobian is $r^{p-1}$ when decomposing a vector (supported on $\mathbbm R^p$) to a radius (supported in $\mathbbm R_+$) and a unit direction (supported on the $p-1$ unit sphere, $\mathbbm S_{p-1}$), so the reciprocal Jacobian is  $|(\zbf^T\zbf)^{\frac{p-1}{2}}|^{-1} = (\zbf^T\zbf)^{\frac{1-p}{2}}$. Thus,
\begin{align*}
p(\zbf) &= \frac{\Gamma(p/2)}{\pi^{p/2}B(a,b)} (\zbf^T\zbf)^{a-1/2} \left( 1+\zbf^T\zbf \right)^{-a-b} \times  (\zbf^T\zbf)^{\frac{1-p}{2}}\\
&= \frac{\Gamma(p/2)}{\pi^{p/2}B(a,b)} (\zbf^T\zbf)^{a-p/2} \left( 1+\zbf^T\zbf \right)^{-a-b}
\end{align*} 
Finally, make the transformation $\betabf = \V\D^{-1/2}\zbf \sigma $, so that $\zbf^T\zbf = \betabf^T\Xbf^T\Xbf\betabf /(\sigma^2n)$, where $\V\D\V^T$ is the eigendecomposition of $\X^T\X/n=\Sigmabf_\X$:
\begin{align*}
p(\betabf \mid \sigma^2) &= \frac{\Gamma(p/2)}{\pi^{p/2}B(a,b)} \left(\betabf^T\Xbf^T\Xbf\betabf /(\sigma^2n)\right)^{a-p/2} \left( 1+\betabf^T\Xbf^T\Xbf\betabf /(\sigma^2n) \right)^{-a-b} 
|\V\D^{-1/2}\sigma|^{-1}\\
&= \frac{\Gamma\left( p/2\right)|\Sigmabf_\X|^{1/2} }{B(a,b)~\pi^{p/2}} \left(\sigma^2 \right)^{-a} \left( \betabf^T \Sigmabf_\X \betabf \right)^{a-p/2} \left(1 + \betabf^T \Sigmabf_\X \betabf /\sigma^2\right)^{-(a+b)}. 
\end{align*} 
\end{proof}

%%% Proposition 2 Proof
\begin{proof} [\textbf{Proof of Proposition \ref{gibbs.uoe}}]
Mixture of normals representation for $a \leq p/2$.  
\begin{align*}
\text{Let } \betabf\mid z, w, \Sigmabf &\sim N\left(0, z w \sigma^2 \left(\X^T\X\right)^{-1} \right), \\
z &\sim \text{Inverse-Gamma}\left(  b, n/2\right), \text{ and} \\
w &\sim \text{ beta}\left( a, p/2-a \right).
\end{align*}

% Then, $p(\betabf|\Sigma) = \frac{\Gamma\left( p/2\right)|\Sigma_X|^{1/2} }{B(a,b) \pi^{p/2}} \left(\sigma^2\right)^a \left( \betabf' \Sigma_X \betabf \right)^{a-p/2} \left(1 + \betabf' \Sigma_X \betabf \right/\sigma^2)^{-(a+b)}$,
Define $\theta = \frac{1-w}{w}$, so $\theta \sim $ BP($p/2-a$, $a$). To simplify notation let $\Sigmabf = \X^T\X /\sigma^2$ and $\Sigmabf_\X = \X^T\X/n$.
Then we have 
\begin{align*}
p(\betabf, \theta, z \mid  \Sigmabf) =& \, (2\pi)^{-p/2} z^{-p/2} (1+\theta)^{p/2} |\Sigmabf|^{1/2} \exp \left\{ -\frac{1+\theta}{2z}\betabf^T\Sigmabf\betabf \right\}\\
&\times \frac{n^b 2^{-b}}{\Gamma(b)}  z^{-b-1} \exp\left\{ -n/2z \right\}
\frac{\Gamma(p/2)}{\Gamma(a)\Gamma(p/2-a)} \theta^{p/2-a-1}(1+\theta)^{-p/2} \\
=& \, \frac{2^{-p/2-b}n^{b}|\Sigmabf|^{1/2}\Gamma(p/2)}{\pi^{p/2}\Gamma(a)\Gamma(b)\Gamma(p/2-a)} 
z^{-p/2-b-1} \theta^{p/2-a-1} \exp \left\{ -\frac{n+\betabf^T\Sigmabf\betabf + \theta\betabf^T\Sigmabf\betabf}{2z} \right\} , 
\end{align*}
and 
\begin{align*}
p(\betabf,z\mid \Sigmabf) =& \, \frac{2^{-\frac{p}{2}-b}n^{b} |\Sigmabf|^{1/2} \Gamma(\frac{p}{2})}{\pi^{\frac{p}{2}}\Gamma(a)\Gamma(b)\Gamma(\frac{p}{2}-a)} z^{-\frac{p}{2}-b-1} \exp \left\{ -\frac{n+\betabf^T\Sigmabf\betabf}{2z} \right\}
\int \theta^{\frac{p}{2}-a-1} \exp \left\{ -\theta \frac{\betabf^T\Sigmabf\betabf}{2z} \right\} d\theta\\
=& \, \frac{2^{-a-b} n^{b}|\Sigmabf|^{1/2}\Gamma(\frac{p}{2})}{\pi^{\frac{p}{2}}\Gamma(a)\Gamma(b)} z^{-a-b-1} \exp \left\{ -\frac{n+\betabf^T\Sigmabf\betabf}{2z} \right\} \left( \betabf^T\Sigmabf\betabf \right)^{a-\frac{p}{2}}. 
\end{align*}
Hence then
\begin{align*}
p(\betabf\mid \Sigmabf) =& \, \frac{2^{-a-b}n^b|\Sigmabf|^{1/2}\Gamma(p/2)}{\pi^{p/2}\Gamma(a)\Gamma(b)} \left( \betabf^T\Sigmabf\betabf \right)^{a-p/2}
\int z^{-a-b-1} \exp \left\{ -\frac{n+\betabf^T\Sigmabf\betabf}{2z} \right\} dz\\
=& \, \frac{2^{-a-b}n^b|\Sigmabf|^{1/2}\Gamma(p/2)}{\pi^{p/2}\Gamma(a)\Gamma(b)} \left( \betabf^T\Sigmabf\betabf \right)^{a-p/2}
{ \left(  n+\betabf^T\Sigmabf\betabf \right)}^{-a-b} 2^{a+b} \Gamma(a+b)\\
=& \, \frac{|n\Sigmabf|^{1/2}\Gamma(p/2)}{\pi^{p/2}B(a,b)} \left( \betabf^T\Sigmabf\betabf/n \right)^{a-p/2} { \left(  1+\betabf^T\Sigmabf\betabf/n \right)}^{-(a+b)}\\
=& \frac{\Gamma\left(p/2\right)|\Sigmabf_\X|^{1/2} }{B(a,b) \pi^{p/2}} \left(\sigma^2\right)^{-a} \left( \betabf^T \Sigmabf_\X \betabf \right)^{a-p/2} \left(1 + \betabf^T \Sigmabf_\X \betabf \right/\sigma^2)^{-(a+b)}. 
\end{align*}
\end{proof}

%%% Proposition 3 Proof
\begin{proof} [\textbf{Proof of Proposition \ref{prop.uoe.equiv}}]
If we let $\Lambdabf^{-1} = \Xbf^T\Xbf = n\V\D\V^T$, then 
\begin{align*}
p(\gammabf) &= C_{X, \Lambda} \exp\left\{ -\gammabf^T \D^{-1/2}\V^T\V\D\V^T\V \D^{-1/2} \gammabf /2 \right\} \ind\left\{ \gammabf^T\gammabf = 1 \right\} = \frac{1}{C} \, \ind\left\{ \gammabf^T\gammabf = 1 \right\}
\end{align*} 
where the constant $C = \int \ind\left\{ \gammabf^T\gammabf = 1 \right\}d\gammabf = \frac{2\pi^{p/2}}{\Gamma(p/2)}$, the surface area of the a $p-1$ dimensional unit sphere. 
The rest of the proof is identical to that of Proposition 1 deriving the distribution in (\ref{eq:priorbetabf}).
It is clear that $\betabf$ is uniform given the ellipsoid, because it's an elliptical distribution, but it also comes from putting a uniform distribution on $\gammabf$. 
\end{proof}

%-------------------------------------------------------------------------------
\begin{proof}[ \textbf{Proof of  	Proposition \ref{proposition_1}}]
	%-------------------------------------------------------------------------------
	The proposition follows from below derivations.  For $\omega>0$, 
	\begin{eqnarray*}
		\pi(\omega) &=& \int_{ 0}^\infty \pi(\omega\mid\xi) \pi(\xi) d\xi = \int_{ 0}^\infty \frac{\xi^a}{\Gamma(a)} \omega^{a-1}e^{-\xi \omega} 
		\frac{1}{\Gamma(b)} \xi^{b-1}e^{-\xi} d\xi \\ 
		&=&    \frac{1}{\Gamma(a) \Gamma(b)} \omega^{a-1} \int_{ 0}^\infty  \xi^{a+b-1} e^{-(1+\omega)\xi} d\xi \\
		&=&  \frac{1}{\Gamma(a) \Gamma(b)} \omega^{a-1} \frac{\Gamma(a+b)}{(1+\omega)^{a+b}} \\
		&=& \frac{\Gamma(a+b)}{\Gamma(a) \Gamma(b)}  \frac{\omega^{a-1}}{(1+\omega)^{a+b}} . 
	\end{eqnarray*}

\end{proof}

%-------------------------------------------------------------------------------  	
\begin{proof} [\textbf{Proof of Proposition \ref{proposition_tauphi_j} }] 
	%-------------------------------------------------------------------------------
	The proposition follows from {Lemma IV.3 of \cite{zhou2015negative}}: 
	Suppose $y$ and $(y_1,\cdots,y_K)$ are independent with $y\sim\text{Ga}(\phi, \xi)$, 
	and $(y_1,\cdots,y_K)\sim\text{Dir}(\phi p_1,\cdots,\phi p_K)$, where $\sum_{k=1}^Kp_k=1$. Let $x_k=yy_k$, then  $x_k\sim\text{Ga}(\phi p_k, \xi)$ independently for $ k=1,\cdots, K $.  
\end{proof}

%-------------------------------------------------------------------------------
\begin{proof} [\textbf{Proof of Proposition \ref{proposition_marginal density}} ] 
	%-------------------------------------------------------------------------------
	The marginal density of $\beta$ for the  R2-D2  prior is 
	\begin{eqnarray} \label{eq_marginal_DH}
	&& \pi_{\text{R2-D2}}(\beta)   \nonumber  \\
	& = & \frac{ \Gamma(a_\pi+b)}{\Gamma(a_\pi) \Gamma(b)}  \int_{0}^\infty \frac{1}{2 (\lambda/2)^{1/2}}\exp\{-\frac{|\beta|}{(\lambda/2)^{1/2}}\} \frac{\lambda^{a_\pi-1}}{(1+\lambda)^{a_\pi+b}} \,d\lambda  \nonumber  \label{eq_R2D2marginal density 1st} \\
	% 	& = & \frac{2^{a_\pi}\Gamma(a_\pi+b)}{\Gamma(a_\pi) \Gamma(b)}  \int_{x=0}^\infty \frac{1}{2\sqrt\bm{x}}\exp\{-\frac{|\beta|}{\sqrt\bm{x}}\} \frac{x^{a_\pi-1}}{(1+2x)^{a_\pi+b}} \,dx \nonumber  \\
	% 	& = & \frac{2^{a_\pi}\Gamma(a_\pi+b)}{\Gamma(a_\pi)\Gamma(b)}  \int_{t=0}^\infty \frac{1}{2t}\exp\{-\frac{|\beta|}{t}\} \frac{t^{2a_\pi-2}}{(1+2t^2)^{a_\pi+b}} 2t \,dt  \nonumber \\
	%% 	& = & \frac{2^{a_\pi}\Gamma(a_\pi+b)}{\Gamma(a_\pi)\Gamma(b)}  \int_{t=0}^\infty \exp\{-\frac{|\beta|}{t}\} \frac{t^{2a_\pi-2}}{(1+2t^2)^{a_\pi+b}}\, dt     \nonumber \\ 
	% 	& = & \frac{2^{a_\pi}\Gamma(a_\pi+b)}{\Gamma(a_\pi)\Gamma(b)}   \int_{x=0}^\infty   \exp\{- |\beta|x \} \frac{x^{(-2)(a_\pi-1)}}{(1+2x^{-2})^{a_\pi+b}} x^{-2} \, dx \nonumber \\
	&= & \frac{2^{a_\pi}\Gamma(a_\pi+b)}{\Gamma(a_\pi)\Gamma(b)}   \int_{ 0}^\infty   \exp (- |\beta|x )\frac{x^{ 2b}}{(x^2+2)^{a_\pi+b}}  \, dx  \label{eq_mariginal density_integral} . 
	\end{eqnarray}
	Let  $\mu = |\beta|$, $\nu=b+1/2$,  $u^2=2$, and $\rho = 1-a_\pi-b$,    since  $|\text{arg } u|<{\pi}/ {2}$, Re$\mu  >0$, and Re$\nu>0$,    so  we have 
	\begin{eqnarray}
	\pi_{\text{R2-D2}}(\beta)  
	&=& \frac{2^{a_\pi}\Gamma(a_\pi+b)}{\Gamma(a_\pi) \Gamma(b)}  \int_{  0}^\infty  \exp(- \mu x ) x^{2\nu-1} (x^2+u^2)^{\rho-1} \, dx  \nonumber\\
	&=& \frac{2^{a_\pi}\Gamma(a_\pi+b)}{\Gamma(a_\pi)\Gamma(b)} \frac{u^{2\nu+2\rho-2}}{2\pi^{1/2} \Gamma(1-\rho)}G^{3,1}_{1,3}\left( \frac{\mu^2 u^2}{4} \left\vert_{1-\rho-\nu,0,\frac{1}{2}}^{1-\nu}\right.
	\right) \nonumber     \\
	&=&  \frac{2^{a_\pi}\Gamma(a_\pi+b)}{\Gamma(a_\pi)\Gamma(b)} \frac{2^{1/2-a_\pi}}{2\pi^{1/2} \Gamma(a_\pi+b)}  G^{3,1}_{1,3}\left( \frac{\beta^2}{2} \left\vert_{a_\pi-\frac{1}{2},0,\frac{1}{2}}^{\frac{1}{2}-b}\right. 	\right)   \nonumber  \\
	&=&   \frac{1}{ (2\pi)^{1/2}\Gamma(a_\pi)\Gamma(b) }  G^{3,1}_{1,3}\left( \frac{\beta^2}{2} \left\vert_{a_\pi-\frac{1}{2},0,\frac{1}{2}}^{\frac{1}{2}-b}\right.   \right)  \nonumber \\
	&=& \frac{1}{ (2\pi)^{1/2}\Gamma(a_\pi)\Gamma(b) }   G^{1,3}_{3,1}\left( \frac{2}{\beta^2} \left\vert_{ \frac{1}{2}+b }^{\frac{3}{2} -  a_\pi,1,\frac{1}{2}}\right.
	\right)
	\nonumber  
	\end{eqnarray}
	where $G(.)$ denotes the Meijer  G-Function,  the second equality   follows from equation 3.389.2 in \cite{zwillinger2014table}, 
	and the last  equality  follows from 
	{16.19.1} in \cite{NISTDLMF}. 
	Proposition \ref{proposition_marginal density}  follows. 
\end{proof}

%-------------------------------------------------------------------------------
\begin{proof} [\textbf{Proof of Theorem \ref{theorem_tail properties}}]
	%-------------------------------------------------------------------------------
	For the proof of Theorem \ref{theorem_tail properties}, we will use the following lemma found in \cite{miller2006applied}. 
	\begin{lemma} (Watson's Lemma)
		Suppose  $F(s) = \int_{ 0}^\infty e^{-st} f(t) \, dt$,   $f(t) = t^\alpha g(t)$  where $g(t)$ has an infinite number of derivatives in the neighborhood of $t=0$,   with $g(0)\neq 0$, and $\alpha>-1$. 
		Suppose  $|f(t)| < Ke^{ct}$ for any $t\in (0,\infty)$, where $K$ and $c$ are independent of $t$. Then, 
		for $s>0$ and $s\rightarrow\infty$, 
		\[
		F(s) = \sum\limits_{k=0}^n \frac{g^{(k)}(0)} {k!}   \frac{\Gamma(\alpha+k+1)}{s^{\alpha+k+1}} + O(\frac{1}{s^{\alpha+n+2}}). 
		\]
	\end{lemma}
	According to  equation (\ref{eq_mariginal density_integral})  in the proof of Proposition \ref{proposition_marginal density},  we denote 
	$F(|\beta|) \equiv	\pi_{\text{R2-D2}}(\beta) $, as follows, 
	\[
	F(|\beta|) = \frac{2^{a_\pi}\Gamma(a_\pi+b)}{\Gamma(a_\pi)\Gamma(b)}   \int_{ 0}^\infty   \exp(- |\beta|x ) \frac{x^{ 2b}}{(x^2+2)^{a_\pi+b}}  \, dx  =   \int_{ 0}^\infty   e^{- |\beta|x} f(x)  \, dx , 
	\]
	where  $f(t) =C^\ast  {t^{ 2b}} / {(t^2+2)^{a_\pi+b}}  \equiv t^{ 2b} g(t) $, 
	$C^\ast = {2^{a_\pi}\Gamma(a_\pi+b)}/ \{\Gamma(a_\pi)\Gamma(b)\}$, and $g(t) = C^\ast(t^2+2)^{-a_\pi-b}$ with $g(t)$ has an infinite number of derivatives in the neighborhood of $t=0$,   with $g(0)\neq 0$.   So the marginal density of R2-D2 prior  is the Laplace transforms of $f(\cdot)$.  By Watson's Lemma,  since $|f(t)| < Ke^{ct}$ for any $t\in (0,\infty)$, where $K$ and $c$ are independent of $t$, then as $|\beta|\rightarrow \infty$, 
	%
	% \[
	% f(t) = C^\ast t^{2b}  \sum\limits_{k=0}^\infty \frac{g^{(k)}(0)}{k!}t^k,
	% \]
	% where $g(t) = (t^2+2)^{-a_\pi-b}$.  The conditions in the Watson's Lemma are satisfied, so 
	\begin{eqnarray*}
		F(|\beta|) &= &  \sum\limits_{k=0}^n  \frac{g^{(k)}(0)}{k!} \frac{\Gamma(2b+k+1)}{|\beta|^{2b+k+1}} + O(\frac{1}{|\beta|^{2b+n+2}}) , 
	\end{eqnarray*} and setting $n=2$ gives
	\begin{eqnarray}
		F(|\beta|)&=& C^\ast\left\{  \frac{\Gamma(2b+1)}{ 2^{a_\pi+b} |\beta|^{2b+1}} + 0\frac{\Gamma(2b+2)}{|\beta|^{2b+2}}  -  \frac{a_\pi+b)\Gamma(2b+3)}{ 2^{a_\pi+b}|\beta|^{2b+3}} \right\}  
		+ O(\frac{1}{|\beta|^{2b+4}}) \nonumber \\
		&=& C^\ast  2^{-a_\pi-b} \left\{ \frac{\Gamma(2b+1)}{|\beta|^{2b+1}} -(a_\pi+b)\frac{\Gamma(2b+3)}{|\beta|^{2b+3}} \right\}  + O(\frac{1}{|\beta|^{2b+4}})  \label{eq-temp}    \label{eq_tail_detail}\\
		&=& O(\frac{1}{|\beta|^{2b+1} }).  \nonumber 
	\end{eqnarray}
	Hence,  when $b < 1/2$,   as $|\beta|\rightarrow \infty$,   we have 
	\begin{eqnarray*}
		&& \frac{\pi_{\text{R2-D2}}(\beta)}{\frac{1}{\beta^2}}
		= \frac{C^\ast}{  2^{a_\pi+b}} \left\{ \frac{\Gamma(2b+1)}{|\beta|^{2b-1}} -(a_\pi+b)\frac{\Gamma(2b+3)}{|\beta|^{2b+1}} + O(\frac{1}{|\beta|^{2b+2}}) \right\}\rightarrow \infty . 
	\end{eqnarray*}
\end{proof}

%-------------------------------------------------------------------------------
\begin{proof} [\textbf{Proof of Theorem \ref{theorem_tail properties Pareto}}]
	It is obvious   based on the marginal density of the generalized double Pareto prior. 
\end{proof}

%-------------------------------------------------------------------------------
\begin{proof} [\textbf{Proof of Theorem \ref{theorem_tail properties DL}}]	 %-------------------------------------------------------------------------------
	According to %\href{http://dlmf.nist.gov/10.25}
	{10.25.3} in \cite{NISTDLMF}, when both $\nu$ and $z$ are real,  if $z\rightarrow\infty$, then   $K_\nu (z) \approx  \pi^{1/2} (2z)^{-1/2}e^{-z}$.    Then as $|\beta| \rightarrow \infty$, the  marginal density of the Dirichlet-Laplace prior given in \cite{bhattacharya2015dirichlet}  satisfies
	\begin{eqnarray*}
		\pi_{\text{DL}}(\beta) 
		&=&  \frac{1}{2^{(1+a^\ast)/2}\Gamma(a^\ast)}  |\beta|^{(a^\ast-1)/2} K_{1-a^\ast}((2|\beta|)^{1/2}) \\
		&\approx &  \frac{1}{2^{(1+a^\ast)/2}\Gamma(a^\ast)}  |\beta|^{(a^\ast-1)/2}  \pi^{1/2}  2^{-3/4}  |\beta|^{-1/4} \exp \{- \sqrt{2|\beta|}\}  \\
		&=&  C_0 |\beta|^{a^\ast/2-3/4} \exp \{- \sqrt{2|\beta|}\}  
		= O(\frac{|\beta|^{a^\ast/2-3/4}}{ \exp\{ { \sqrt{2|\beta|}\} } }) , 
	\end{eqnarray*}
	where $C_0 = { {\pi^{1/2} 2^{-3/4}}} / \{2^{(1+a^\ast)/2}\Gamma(a^\ast)\}$ is a constant value. 
	% Then as $|\beta| \rightarrow \infty$,  we have
	% $\pi_{\text{DL}}(\beta)  \rightarrow 0 $ for any nonnegative value of $a^\ast$. 
	Furthermore,  as $|\beta| \rightarrow \infty$,
	\[
	\frac{\pi_{\text{DL}}(\beta) }{1/\beta^2} \approx  C_0 |\beta|^{a^\ast/2+5/4}  \exp \{- \sqrt{2|\beta|}\}   \rightarrow 0 . 
	\]
\end{proof}

%-------------------------------------------------------------------------------
\begin{proof} [\textbf{Proof of Theorem \ref{theorem_center properties R2-D2}}]
	%-------------------------------------------------------------------------------
	For the proof of Theorem \ref{theorem_center properties R2-D2}, we use the following lemma from \cite{fields1972asymptotic}.  Some useful notations used in the below proof:
	Denote $a_P = (a_1,\cdots, a_p)$, as  a vector, similarly,  $b_Q = (b_1,\cdots, b_q)$,  $c_M = (c_1,\cdots, c_m)$, and so on. Let 
	$\Gamma_n(c_P-t) = \prod_{k=n+1}^p \Gamma(c_k-t)$,  with 	$\Gamma_n(c_P-t) =1$ when $n=p$, 
	$\Gamma(c_M-t) = \Gamma_0(c_M-t) =\prod_{k=1}^m \Gamma(c_k-t)$,  
	$\Gamma^\ast(a_i-a_N) = \prod_{k=1;k\neq i}^n  \Gamma(a_i-a_k)$, and 
	\[	{}_pF_q\left(  ^{a_P}_{b_Q} \mid  w
	\right) = \sum\limits_{k=0}^{\infty} \frac{\Gamma(a_P+k) \Gamma(b_Q)}{\Gamma(b_Q+k)\Gamma(a_P)} \frac{w^k}{k!} 
	= \sum\limits_{k=0}^{\infty} \frac{\prod\limits_{j=1}^p\Gamma(a_j+k) \prod\limits_{j=1}^q\Gamma(b_j)}{\prod\limits_{j=1}^q\Gamma(b_j+k)\prod\limits_{j=1}^p\Gamma(a_j)} \frac{w^k}{k!}.
	\]
	\begin{lemma} \label{lemma_meijer g approximation}
		(Theorem 1 in \cite{fields1972asymptotic})
		Given  (i) $0\leq m \leq q$, $0\leq n \leq p$; (ii) $a_i - b_k $ is not a positive integer for $j=1, \cdots,  p$ and $k=1,.s , q$;   (iii) $a_i-a_k$ is not an integer for $i,k = 1,\cdots,  p$,  and $i\neq k$; and (iv)  $q< p$  or $q=p$ and $|z|>1$, 
		we have 
		\begin{eqnarray*}
			G_{p,q}^{m,n} 
			\left( z \left\vert^{  a_1, \cdots,a_p}_{b_1, \cdots, b_q} \right.
			\right) 
			=  \sum\limits_{i=1}^{n}      \frac{\Gamma^\ast (a_i-a_N) \Gamma(1+b_M-a_i) z^{-1+a_i}} {\Gamma_n(1+a_P-a_i)\Gamma_m(a_i-b_Q)} 
			{}_{q+1}F_p
			\left(  _{1+a_P-a_i}^{1,1+b_Q-a_i} 
			\left\vert\frac{(-1)^{q-m-n}}{z} \right. 
			\right). 
		\end{eqnarray*}
	\end{lemma}
	
	Now to prove Theorem \ref{theorem_center properties R2-D2}, we have from Proposition \ref{proposition_marginal density} that, 
	the marginal density of  the R2-D2  prior has $\pi_{\text{R2-D2}}(\beta_j) =   (2\pi)^{-1/2}\{ \Gamma(a_\pi)\Gamma(b) \}^{-1}G^{m,n}_{p,q} \left(z  | . \right)$ with $m=1, \ n=3,\ p=3,\ q=1,\ a_1=3/2-a_\pi,\ a_2=1, \ a_3=1/2$,    $b_1 =1/2+b$, and $z=2/{\beta^2}$. 
	Conditions (i)-(iv)  in   Lemma  \ref{lemma_meijer g approximation} are  satisfied 
	for $|\beta|$ near 0, since $0<a_\pi<1/2$. 
	Denote 	 
	{  \begin{eqnarray*}
			C_0^\ast &=& { (2\pi)^{-1/2}(\Gamma(a_\pi)\Gamma(b))^{-1} },  \\
			C_1^\ast &=&C_0^\ast \Gamma(\frac{1}{2}-a_\pi)\Gamma(1-a_\pi)\Gamma(a_\pi) \Gamma(\frac{1}{2}+a_\pi)  >0 , \\
			C_2^\ast &=&C_0^\ast \Gamma(a_\pi-\frac{1}{2}) \Gamma(\frac{1}{2})\Gamma(\frac{1}{2}) \Gamma(\frac{3}{2}-a_\pi) <0  , \\
			C_3^\ast &=&C_0^\ast \Gamma(a_\pi-1)\Gamma(-\frac{1}{2})\Gamma(\frac{3}{2})\Gamma(2-a_\pi) > 0 , \\
			%		\end{eqnarray*}
			%		\begin{eqnarray*}
			U_1(\beta^2) &=& 	
			%			\sum\limits_{k=0}^{\infty} \frac{\Gamma(a_\pi+b+k) }{\Gamma(\frac{1}{2}+a_\pi+k)\Gamma(a_\pi+k) } \frac{(-1)^k(\frac{\beta^2}{2})^{k+a_\pi-1/2}} {k!}  \equiv
			\sum\limits_{k=0}^{\infty}(-1)^k  u_1(k, \beta^2),  \ 		u_1(k, \beta^2) = \frac{\Gamma(a_\pi+b+k) }{\Gamma(\frac{1}{2}+a_\pi+k)\Gamma(a_\pi+k) } \frac{ (\frac{\beta^2}{2})^{k+a_\pi-1/2}} {k!}  , \\
			U_2(\beta^2) &=&  	\sum\limits_{k=0}^{\infty} (-1)^k u_2(k, \beta^2),    \ 
			u_2(k, \beta^2)= 	\frac{\Gamma(\frac{1}{2}+b+k) }{\Gamma(\frac{3}{2}-a_\pi+k)\Gamma(\frac{1}{2}+k)} \frac{(\frac{\beta^2}{2})^k}{k!} ,  \\ 
			\text{ and }
			U_3(\beta^2) &=& 	\sum\limits_{k=0}^{\infty} (-1)^ku_3(k,\bm\beta^2), \  u_3(k,\bm\beta^2) = \frac{\Gamma(1+b+k)}{\Gamma(2-a_\pi+k)\Gamma(\frac{3}{2}+k)}
			\frac{(\frac{\beta^2}{2})^{k+1/2}}{k!}  . 
	\end{eqnarray*}}
	Then  
	{ \begin{eqnarray*}
			&&\pi_{\text{R2-D2}} (\beta) =C_0^\ast    G^{1,3}_{3,1}\left( \frac{2}{\beta^2} \left\vert_{ \frac{1}{2}+b }^{\frac{3}{2} -  a_\pi,1,\frac{1}{2}}\right.
			\right) \\
			%  					&=&  C_0^\ast  \sum\limits_{i=1}^{3}   \frac{\Gamma^\ast (a_i-a_\pi_N) \Gamma(1+b_1-a_\pi_i)}{\Gamma_3(1+a_\pi_P-a_\pi_i)\Gamma_1(a_i-b_1)} (\frac{2}{\beta^2})^{-1+a_\pi_i} 
			%  					{}_{q+1}F_p
			%  					\left(  _{1+a_\pi_P-a_\pi_i}^{1,1+b_1-a_\pi_i} 
			%  					\left\vert\frac{-\beta^2}{2} \right. \right) \\
			%  					%  
			&=&   C_0^\ast  \sum\limits_{i=1}^{3}   \frac{\Gamma^\ast (a_i-a_N) \Gamma(1+b_M-a_i)}  {\Gamma_3(1+a_P-a_i)\Gamma_1(a_i-b_Q)}(\frac{2}{\beta^2})^{-1+a_i} 
			{}_{2}F_3
			\left(  _{1+a_P-a_i}^{1,1+b_Q-a_i} 
			\left\vert\frac{-\beta^2}{2} \right. \right) \\   
			&=&    C_0^\ast  \sum\limits_{i=1}^{3}
			\frac{ \prod\limits_{k=1;k\neq i}^3  \Gamma(a_i-a_k)  \Gamma(1+b_1-a_i)}
			{\prod\limits_{k=3+1}^3 \Gamma(1+a_k-a_i) \prod\limits_{k=1+1}^1 \Gamma(a_i-b_k)}(\frac{2}{\beta^2})^{-1+a_i} 
			{}_{2}F_3
			\left(  _{1+a_P-a_i}^{1,1+b_1-a_i} 
			\left\vert\frac{-\beta^2}{2} \right. \right) \\  
			%%%%%
			&=&    C_0^\ast  \sum\limits_{i=1}^{3} \left\{
			\frac{ \prod\limits_{k=1;k\neq i}^3  \Gamma(a_i-a_k)  \Gamma(1+b_1-a_i)}
			{1}(\frac{2}{\beta^2})^{-1+a_i}  \times  \right. \\
			&&	\left. 
			\sum\limits_{k=0}^{\infty} \frac{\Gamma(1+k)\Gamma(1+b_1 - a_i +k) \prod\limits_{j=1}^3 \Gamma(1+a_j-a_i)}
			{\prod\limits_{j=1}^3 \Gamma(1+a_j-a_i+k)  \Gamma(1) \Gamma(1+b_1-a_i) } \frac{(\frac{-\beta^2}{2})^k}{k!}\right\}
			\\    
			&=&  C_0^\ast  \times \\
			&& \left\{  
			\Gamma(\frac{1}{2}-a_\pi)\Gamma(1-a_\pi)   
			\sum\limits_{k=0}^{\infty} \frac{\Gamma(a_\pi+b+k) \Gamma(\frac{1}{2}+a_\pi)\Gamma(a_\pi)}{\Gamma(\frac{1}{2}+a_\pi+k)\Gamma(a_\pi+k) } \frac{(-1)^k(\frac{\beta^2}{2})^{k+a_\pi-1/2}} {k!} \right. \\
			&& +   \Gamma(a_\pi-\frac{1}{2}) \Gamma(\frac{1}{2}) \Gamma(\frac{1}{2})
			\sum\limits_{k=0}^{\infty} \frac{\Gamma(\frac{1}{2}+b+k) \Gamma(\frac{3}{2}-a_\pi)}{\Gamma(\frac{3}{2}-a_\pi+k)\Gamma(\frac{1}{2}+k)} \frac{(-1)^k(\frac{\beta^2}{2})^k}{k!} \\
			&& \left.  + \Gamma(a_\pi-1)\Gamma(-\frac{1}{2})   \Gamma(\frac{3}{2})
			\sum\limits_{k=0}^{\infty}  \frac{\Gamma(1+b+k)\Gamma(2-a_\pi)}{\Gamma(2-a_\pi+k)\Gamma(\frac{3}{2}+k)}
			\frac{(-1)^k(\frac{\beta^2}{2})^{k+1/2}}{k!}  \right\}\\    
			%%%
			&\equiv&  C_0^\ast  \times 
			\left\{	\Gamma(\frac{1}{2}-a_\pi)\Gamma(1-a_\pi)\Gamma(a_\pi) \Gamma(\frac{1}{2}+a_\pi)U_1(\beta^2) +   \right. \\
			&& \left. 
			\Gamma(a_\pi-\frac{1}{2}) \Gamma(\frac{1}{2})\Gamma(\frac{1}{2}) \Gamma(\frac{3}{2}-a_\pi) U_2(\beta^2)   
			+  \Gamma(a_\pi-1)\Gamma(-\frac{1}{2})\Gamma(\frac{3}{2})\Gamma(2-a_\pi)   U_3(\beta^2)  \right\}	 \\
			&\equiv& C_1^\ast U_1(\beta^2) +  C_2^\ast U_2(\beta^2) +  C_3^\ast U_3(\beta^2) . 
			% &\rightarrow& C_0^\ast  \left\{  \Gamma(\frac{1}{2}-a_\pi)\Gamma(1-a_\pi)  (\frac{2}{\beta^2})^{\frac{1}{2}-a_\pi}
			% 		\frac{\Gamma(a_\pi+b) \Gamma(\frac{1}{2}+a_\pi)\Gamma(a_\pi)}{\Gamma(\frac{1}{2}+a_\pi)\Gamma(a_\pi) } 
			% 		+   \Gamma(a_\pi-\frac{1}{2}) \Gamma(\frac{1}{2}) 
			% 		\frac{\Gamma(\frac{1}{2}+b) \Gamma(\frac{3}{2}-a_\pi)\Gamma(\frac{1}{2})}{\Gamma(\frac{3}{2}-a_\pi)\Gamma(\frac{1}{2})}    \right\} \\   
			% 		&=&  C_0^\ast  \left\{  \Gamma(\frac{1}{2}-a_\pi)\Gamma(1-a_\pi)   \Gamma(a_\pi+b) (\frac{2}{\beta^2})^{\frac{1}{2}-a_\pi}
			% 		+   \Gamma(a_\pi-\frac{1}{2})\Gamma(\frac{1}{2}) \Gamma(\frac{1}{2}+b)    \right\}   .
			% % % % % % % % % % % % % % % % % % % % % % % % % % % % % % % % % % % % % % % %    
		\end{eqnarray*} 
	}
	For fixed  $\beta$ near the neighborhood of zero,  $u_1(k,\beta^2)$, 	$u_2(k,\beta^2)$, 	and $u_3(k,\beta^2)$ are all monotone decreasing, and converge to zero as $k\rightarrow\infty$. Thus, by alternating series test,  $U_1(\beta^2$),  $U_2(\beta^2)$, and $U_3(\beta^2)$ all converge.  
	Also, we have
	{  \[
		C_0|\beta|^{2a_\pi-1 } - C_1|\beta|^{2a_\pi+1}  =u_1(0,\beta^2) -  u_1(1,\beta^2)   \leq U_1(\beta^2)   \leq  u_1(0,\beta^2) = C_0 |\beta|^{2a_\pi-1 }  
		\]
		\[
		C_2 - C_3|\beta|^2  = u_2(0,\beta^2) -  u_2(1,\beta^2)   \leq U_2(\beta^2) \leq  u_2(0,\beta^2) =C_2
		\]
		\[
		C_4|\beta| -  C_5|\beta|^3   = u_3(0,\beta^2) -  u_3(1,\beta^2)   \leq U_3(\beta^2)  \leq  u_3(0,\beta^2) = C_4|\beta| 
		\]}
	where $C_0, \ C_1,\  C_2, \  C_3$, and   $C_4$ are all positive constants. 
	So  given that $|\beta|$ in the neighborhood of zero and  $a_\pi\in (0, \frac{1}{2}) $, 
	$
	C_1^\ast (C_0|\beta|^{2a_\pi-1 } - C_1|\beta|^{2a_\pi+1}) + C_2^\ast C_2 + C_3^\ast(C_4|\beta| -  C_5|\beta|^3 ) \leq  \pi_{\text{R2-D2}}(\beta)  \leq
	C_1^\ast C_0 |\beta|^{2a_\pi-1 }    +   C_2^\ast(C_2 - C_3|\beta|^2) + C_3^\ast C_4|\beta| , 
	$
	% {\small 	\[
	%  C_0|\beta|^{2a_\pi-1 }) - O(|\beta|^{2a_\pi+1}) + O(1) - O(|\beta|^2)  + O(|\beta|) -  O(|\beta|^3) 	 \leq \pi_{\text{R2-D2}}(\beta) \leq O(|\beta|^{2a_\pi-1 }) +O(1) +  O(|\beta|)  , 
	% 	\]}
	then 	
	$\pi_{\text{R2-D2}}(\beta)  = O(|\beta|^{2a_\pi-1 })$. 
\end{proof}

%------------------------------------------------------------------------------- 
\begin{proof} [\textbf{Proof of Theorem \ref{theorem_center properties DL}}]
	%-------------------------------------------------------------------------------
	According to  %\href{http://dlmf.nist.gov/10.30}
	{10.30.2} in \cite{NISTDLMF}, when $\nu>0$,  $z\rightarrow0$ and $z$ is   real, $K_\nu (z) \approx   \Gamma(\nu) (z/2)^{-\nu} /2$.
	So  given $0<a^\ast<1$ and  $|\beta|\rightarrow0$, 
	\begin{eqnarray*}
		\pi_{\text{DL}}(\beta)   
		&	= & \frac{|\beta|^{(a^\ast-1)/2} K_{1-a^\ast} ( {(2|\beta|)^{1/2}})}{2^{(1+a^\ast)/2}\Gamma(a^\ast)}   \\
		& \approx  &  \frac{|\beta|^{(a^\ast-1)/2} \frac{1}{2} \Gamma(1-a^\ast) (\frac{(2|\beta|)^{1/2}}{2} )^{a^\ast-1}
		}{2^{(1+a^\ast)/2}\Gamma(a^\ast)}  
		=  C|\beta|^{a^\ast-1} ,   %= O(1/|\beta|^{1-a^\ast})
	\end{eqnarray*}
	where $C= { \Gamma(1-a^\ast) } / {2^{1+a^\ast}\Gamma(a^\ast)} $ is a constant value.  Theorem \ref{theorem_center properties DL} follows then. 
\end{proof}

\begin{proof} [\textbf{Proof of Theorem \ref{theorem_R2-D2_consistency}}]
	%------------------------------------------------------------------------------- 
	The proof of this theorem depends on Theorem 1 in \cite{armagan2013posterior}.   We will restate this theorem in the following Lemma. 
	
%\textbf{	Lemma 1 in \cite{armagan2013posterior}:} Let $\mathcal{B}_n = \{ \beta_n: ||\beta_n-\beta_n^0 || >\varepsilon\}$ where $\varepsilon>0$. To test $H_0: \beta_n = \beta_n^0$ vs $H_1: \beta_n\in \mathcal{B}_n$, we define a test function $\Phi_n(y_n)=I(y_n\in \mathcal C_n)$ where the critical region is $\mathcal C_n :=\{ y_n: ||\hat\beta_n  -\beta_n^0||>\varepsilon/2\}$ and $\hat\beta_n = (X_n^TX_n)^{-1} X_n^T y_n$. Then, under Assumptions  \ref{a_p=o(n)} and \ref{a_eigenvalues}, as $n\rightarrow\infty$, 
%
%(i) $\text{E}_{\beta_n^0}(\Phi_n) \leq \exp\{  -\varepsilon^2 n d^2_{\text{min}} /(16\sigma^2) \}$
%
%(ii)  $\sup_{\beta_n\in \mathcal B_n} \text{E}_{\beta_n} (1-\Phi_n)\leq \exp \{  -\varepsilon^2 nd^2_{\text{min}} /(16\sigma^2) \}$. 

\begin{lemma}\label{lemma_theorem 1 in armagan}
Under Assumptions \ref{a_p=o(n)} and \ref{a_eigenvalues}, the posterior of $\beta_n$ under prior $\pi_n(\beta_n)$ is strongly consistent, that is, for any $\epsilon>0$,  as $n\rightarrow\infty$, 
\[
\text{Pr}_{\bm\beta^0_n}  \left\{ \pi_n ( \beta_n: ||\beta_n - \beta_n^0|| >\epsilon \mid Y_n  )\rightarrow 0 \right\}  =1
\] 
if
\[
\pi_n ( \beta_n: ||\beta_n - \beta_n^0||  < \frac{\Delta}{n^{r/2}}) >\exp(-dn)
\]
for all $0<\Delta<\epsilon^2 d_\text{min} /(48d_\text{max})$ and $0<d<\epsilon^2 d_\text{min}/(32\sigma^2) -3\Delta d_\text{max} /(2\sigma^2)$ and some $r>0$. 
\end{lemma}

%		Denote  the  true set of non-zero coefficients is $\mathcal{A}_n^0=\{j:\beta_{j}^0\neq0, j=1,\cdots, p_n\}$,
	Denote  the 
	estimated set of non-zero coefficients as $\mathcal{A}_n=\{j:\beta_{nj}\neq0, j=1,\cdots, p_n\}$.   
%	 Without loss of generality, $\sigma^2$ is fixed at 1. 
	Given the  R2-D2  prior (\ref{eq_R2-D2 prior, linear model, simple representation}),  we need to calculate the probability assigned to the region  $\{ \bm\beta:|| \bm\beta- \bm\beta^0|| < t_n\}$  where   $t_n =  {\Delta} / {n^{r/2}}$  with $0<\Delta < \epsilon^2 d_\text{min}/ (48 d_\text{max})$. 
	
	{\small \begin{eqnarray*}\label{6}
		&& 	\pi_n( \bm\beta_n:|| \bm\beta_n- \bm\beta^0_n|| < t_n) = 
		\pi_n\left\{ \bm\beta_n: \sum\limits_{j\in \mathcal{A}_n } (\beta_{nj}-\beta_{nj}^0)^2 + \sum\limits_{j\not\in \mathcal{A}_n}\bm\beta_{n j}^2 < t_n^2\right\} \nonumber \\
		&\geq& 
		\pi_n\left\{\bm\beta_{nj}^{j\in\mathcal{A}_n}: \sum\limits_{j\in \mathcal{A}_n }  (\beta_{nj}-\beta_{nj}^0)^2 <\frac{q_n t_n^2}{{p_n} }\right\} \times 
		\pi_n\left\{\bm\beta_{nj}^{j\not\in\mathcal{A}_n} : \sum\limits_{j\not\in \mathcal{A}_n}
		\beta_{nj}^2 < \frac{(p_n-q_n)t_n^2}{p_n } \right\}  \nonumber \\
		&\geq& 
		\left[\prod\limits_{j\in\mathcal{A}_n} \left\{ \pi_n\left( \beta_{nj}:   |\beta_{nj}-\beta_{nj}^0| <\frac{t_n}{{\sqrt{p_n}}}\right)\right\}  \right]
		\pi_n\left(\beta_{nj}^{j\not\in\mathcal{A}_n} :  \beta_{nj}^2 < \frac{ t_n^2}{p_n }   \text{at least for one $j$} \right)     \nonumber \\
		&= & \left[ \prod\limits_{j\in\mathcal{A}_n} \left\{ \pi_n\left( \beta_{nj}^0 -\frac{t_n}{{\sqrt{p_n}} } <   \beta_{nj} < \beta_{nj}^0 +\frac{t_n}{{\sqrt{p_n}} }\right)\right\}  \right] 
	 \left[  1 - \left\{ \pi_n\left(\beta_{nj}^{j\not\in\mathcal{A}_n} :  
		\beta_{nj}^2 \geq  \frac{ t_n^2}{p_n } \right) \right\}^{p_n-q_n} \right]\nonumber \\
				&\geq & \left\{\prod\limits_{j\in\mathcal{A}_n} 
				2\frac{t_n}{{\sqrt{p_n}}} \pi_{\text{R2-D2}} \left(   |\beta_{nj}^0| + \frac{t_n}{{\sqrt{p_n}}} \right) \right\}
				 \left[  1 - 	\left\{ \pi_n\left(\beta_{nj}^{j\not\in\mathcal{A}_n} :  
				|\beta_{nj}|^b \geq \frac{ t_n^b}{p_n^{b/2} } \right)\right\}^{p_n-q_n} \right] \nonumber \\
					&\geq & \left\{\prod\limits_{j\in\mathcal{A}_n} 
				2\frac{t_n}{{\sqrt{p_n}}} \pi_{\text{R2-D2}} \left(    \sup_{j\in\mathcal{A}_n} |\beta_{nj}^0| + \frac{t_n}{{\sqrt{p_n}}} \right) \right\}
				\left[1- 	\left\{    \frac{p_n^{b/2}   \text{E}( |\beta_{nj}|^b)}{ t_n^b} \right\}^{p_n-q_n}  \right ]\nonumber \\
		%
%		&\geq & \left[\prod\limits_{j\in\mathcal{A}_n} \left\{ \pi_n\left(-\sup_{j\in\mathcal{A}_n}|\beta_{nj}^0| -\frac{t_n}{{\sqrt{p_n}} } <   \beta_{nj} < \sup_{j\in\mathcal{A}_n}|\beta_{nj}^0| +\frac{t_n}{{\sqrt{p_n}} }\right)\right\} \right]
%		 \left[  1 - 	\left\{ \pi_n\left(\beta_{nj}^{j\not\in\mathcal{A}_n} :  
%		|\beta_{nj}|^b \geq \frac{ t_n^b}{p_n^{b/2} } \right)\right\}^{p_n-q_n} \right] \nonumber \\
%		%
%		%
%		&\geq & \left[\prod\limits_{j\in\mathcal{A}_n} \left\{ \pi_n\left(-\sup_{j\in\mathcal{A}_n}|\beta_{nj}^0| -\frac{t_n}{{\sqrt{p_n}} } <   \beta_{nj} < \sup_{j\in\mathcal{A}_n}|\beta_{nj}^0| +\frac{t_n}{{\sqrt{p_n}} }\right)\right\}  \right]
%		\left[1- 	\left\{    \frac{p_n^{b/2}   E( |\beta_{nj}|^b)}{ t_n^b} \right\}^{p_n-q_n}  \right ]\nonumber \\
%		%
		&\geq &  \left\{  2\frac{t_n}{{\sqrt{p_n}} } \pi_{\text{R2-D2}}\left(\sup_{j\in\mathcal{A}_n}|\beta_{nj}^0| + \frac{t_n}{{\sqrt{p_n}} } \right) \right\}^{q_n} \times 
		\left[ 	1-   \left\{ \frac{p_n^{b/2}   \text{E}( |\beta_{nj}|^b)}{ t_n^b} \right\}^{p_n-q_n} \right]  , \nonumber  
	\end{eqnarray*}}
	where $\pi_{\text{R2-D2}}(\cdot)$ is the marginal density function of $\beta_{nj}$,   symmetric and decreasing when the support is positive, and the last but one ``$\geq$'' is directly got from Markov's inequality.

	Using the hierarchical form of the R2-D2  prior in (\ref{eq_R2-D2 prior, linear model, simple representation1}),  for any $b>0$,   conditional expectations yield
	\[ 
	\text{E}(|\beta_{nj}|^b) =  \text{E} [ \text{E}   \{ \text{E} (|\beta_{nj}|^b \mid \lambda_j) \mid \xi  \}]
	=  \text{E}_{\xi}  \left[ \text{E} _{\lambda_j \mid  \xi}  \left\{  \frac{\Gamma(b+1)}{({2/\lambda_j})^{b/2} }\mid \xi \right\} \right]
	= \frac{b\Gamma(\frac{b}{2})\Gamma(a_{\pi}+\frac{b}{2}) }{2^{b/2}\Gamma(a_{\pi})}  . 
	\]

	For the R2-D2 prior, from equation (\ref{eq_mariginal density_integral}), it follows that the marginal density is a decreasing function on the positive support.  Then assumptions \ref{a_p=o(n)} -- \ref{a_qplog}, 
	together with  the tail approximation of the marginal density as in the proof of  Theorem \ref{theorem_tail properties}, i.e., equation (\ref{eq_tail_detail}), we have 
	\begin{eqnarray*}
		\pi_\text{R2-D2}(\sup_{j\in\mathcal{A}_n}|\beta_{nj}^0| + \frac{t_n}{{\sqrt{p_n}} } )  
		\geq \pi_\text{R2-D2}(E_n + \frac{t_n}{{\sqrt{p_n}} } ) 
		\geq  \frac{\Gamma(a_{\pi}+b)}{\Gamma(a_{\pi})\Gamma(b)}  2^{-b} \frac{\Gamma(2b+1)}{(E_n  + \frac{\Delta}{n^{r/2}{\sqrt{p_n}} })^{2b+1}}  .
 	\end{eqnarray*}

	Considering  the fact that  $\Gamma(a)=  a^{-1} -\gamma_0+ O(a)$ for $a$ being near zero with $\gamma_0$   the Euler-Mascheroni constant, 
	%		 	Followed by  \href{http://functions.wolfram.com/GammaBetaErf/Gamma/06/ShowAll.html}	{http://functions.wolfram.com/GammaBetaErf/Gamma/06/ShowAll.html},  
	we have 
	{\small \begin{eqnarray*}
		&& \pi_n( \bm\beta_n:|| \bm\beta_n- \bm\beta_n^0|| < \frac{\Delta}{n^{r/2}}) \\
		&\geq & 	\left\{  2\frac{\Delta}{n^{r/2}{\sqrt{p_n}} }
		\frac{\Gamma(a_{\pi}+b)}{\Gamma(a_{\pi})\Gamma(b)}  2^{- b}\frac{\Gamma(2b+1)}{(E_n  + \frac{\Delta}{n^{r/2}{\sqrt{p_n}} })^{2b+1}}  
		\right\}^{q_n}  
		\left[1 -	\left\{ \frac{p_n^{b/2} n^{r b/2} b\Gamma(\frac{b}{2})\Gamma(a_{\pi}+\frac{b}{2})    }{ \Delta^b 2^{b/2}\Gamma(a_{\pi}) } \right\}^{p_n-q_n}  \right]\\
		&\geq & 	\left\{  \frac{ 2 \Delta}{  n^{r/2}{\sqrt{p_n}} } 
		\frac{\Gamma(a_{\pi}+b) a_{\pi}}{ \Gamma(b)}  2^{- b} \frac{\Gamma(2b+1)}{(E_n   + \frac{\Delta}{n^{r/2}{\sqrt{p_n}} })^{2b+1}}  
		\right\}^{q_n}  
		\left[1 -  	\left\{   \frac{p_n^{b/2} n^{r b/2} b\Gamma(\frac{b}{2})\Gamma(a_{\pi}+\frac{b}{2}) a_{\pi}   }{ \Delta^b 2^{b/2}  } \right\}^{p_n-q_n} \right].
	\end{eqnarray*}}
	Taking the negative logarithm of both sides of the above formula, and letting $a_{\pi}= C/( p_n^{b/2} n^{r b/2 }\log n)$,  
	we have 
	\begin{eqnarray*}
		&& -\log \pi_n ( \bm\beta_n:|| \bm\beta_n- \bm\beta_n^0|| < \frac{\Delta}{n^{r/2}}) \nonumber\\
		& \leq & -q_n \log\left\{ \frac{2\Delta C  \Gamma(a_{\pi}+b) 2^{- b}  \Gamma(2b+1) }{ n^{r/2}{\sqrt{p_n}} p_n^{b/2} n^{r b/2 }\log n \Gamma(b)}  \right\}     
		+ q_n ( 2b+1 ) \log (E_n  + \frac{\Delta}{n^{r/2}{\sqrt{p_n}} }) \\
		&& - q_n  \log \left[ 1 -	\left\{  \frac{p_n^{b/2} n^{r b/2} b\Gamma(\frac{b}{2})\Gamma(a_{\pi}+\frac{b}{2}) C   }{ \Delta^b 2^{b/2}  p_n^{b/2} n^{r b/2 }\log n } \right\}^{p_n-q_n} \right] \\
		& = & -q_n \log \left\{   \frac{2 \Delta C \Gamma(a_{\pi}+b) 2^{- b}  \Gamma(2b+1) }{ \Gamma(b)} \right\}   + q_n ( 2b+1 ) \log (E_n   + \frac{\Delta}{n^{r/2}{\sqrt{p_n}} })  \\
		&&	-    q_n  \log\left[ 1-	\left\{  \frac{  b\Gamma(\frac{b}{2})\Gamma(a_{\pi}+\frac{b}{2}) C   }{ \Delta^b 2^{b/2}  \log n } \right\}^{p_n-q_n}  \right]  
		+   q_n\log\log n  \\
		&& +  
		\frac{b+1}{2} q_n   \log p_n +   	\frac{b+1}{2} q_n r  \log n   
	\end{eqnarray*}
Since $\log E_n = O(\log n)$, the dominating term is $O (q_n \log n)$. Hence, if $q_n = o(n/\log n )$, then we have that
	$-\log \pi_n ( \bm\beta_n:|| \bm\beta_n- \bm\beta_n^0|| < {\Delta}/ {n^{r/2}}) < dn$ for all  $0<d<\epsilon^2 d_\text{min}/(32\sigma^2) -3\Delta d_\text{max} /(2\sigma^2)$, so  $\pi_n( \bm\beta_n:|| \bm\beta_n- \bm\beta_n^0||  <  {\Delta}/{n^{r/2}})  > \exp(-dn)$. The posterior consistency is then completed by applying Lemma \ref{lemma_theorem 1 in armagan}. 
\end{proof}

%------------------------------------------------------------------------------- 

\begin{proof} [\textbf{Proof of Theorem \ref{theorem_new theorem}}] 
	Now that we have established the properties of marginal prior for the R2-D2 hierarchical formulation, the proof will now be based on the results similar to Theorems 2.1, 2.2, A.1 and A.2  in \cite{song2017nearly}.  We will restate these theorems in the following lemma. 
	
	\begin{lemma}\label{lemmaB1}
		Consider the linear regression model (\ref{eq_linear_regression}) and suppose the regularity conditions  \ref{new1}-\ref{new4}  hold. Suppose that the prior for $\pi(\beta, \sigma^2)$ is of the form 
		\[
		\pi(\bm\beta\mid \sigma^2)  = \prod\limits_{i=1}^p [g(\beta_i/\sigma) / \sigma ], \ \sigma^2 \sim \text{IG} (a_1, b_1).
		\]
		Denote $\epsilon_n = M\sqrt{q_n (\log p_n)/n}$  where   $M>0$ is sufficiently large. 
		If the density $g(\cdot)$  in the above formula satisfies
		\begin{equation}\label{condition_b1}
		1 - \int_{-k_n}^{k_n} g(\beta)\, d\beta \leq p_n^{-(1+u)},  \ 
		-\log\left(  \inf\limits_{\beta\in[-E_n, E_n]} g(\beta) \right) = O(\log p_n), 
		\end{equation}
		where $u>0$ is a constant and $k_n \asymp \sqrt{q_n (\log p_n)/n}/p_n$, then the following  results hold: 
		\[
		\text{Pr}_{\bm\beta^0} \Big\{\pi (\bm\beta_n: || \bm\beta_n - \bm\beta^0_n || \geq c_1 \sigma^0 \epsilon_n \mid \bm Y_n ) \geq e^{-c_2 n\epsilon_n^2}   \Big\}  \leq e^{-c_3 n\epsilon_n^2}, 
		\]
		\[
		\text{Pr}_{\bm\beta^0} \Big\{ \pi (\bm\beta_n: || \bm\beta_n - \bm\beta^0_n ||_1 \geq c_1 \sigma^0  \sqrt{q_n}\epsilon_n \mid \bm Y_n ) \geq e^{-c_2 n\epsilon_n^2}   \Big\}   \leq e^{-c_3 n\epsilon_n^2}, 
		\]
		\[
		\text{Pr}_{\bm\beta^0} \Big\{ \pi (\bm\beta_n: || \bm X_n\bm\beta_n - \bm X_n \bm\beta^0_n || \geq c_1 \sigma ^0 \sqrt{n}\epsilon_n \mid \bm Y_n ) \geq e^{-c_2 n\epsilon_n^2}   \Big\}  \leq e^{-c_3 n\epsilon_n^2}, 
		\]
		for some constants $c_1, c_2, c_3>0$. 
	\end{lemma}
	
	For the R2-D2 prior in (\ref{eq_R2-D2 prior, linear model, simple representation}),  according to (\ref{eq_mariginal density_integral}) the corresponding $g(\cdot)$ function is 
	\[
	g(\beta)  = \frac{2^{a_{\pi}}\Gamma(a_{\pi}+b)}{\Gamma(a_{\pi})\Gamma(b)}   \int_{ 0}^\infty   \exp (- |\beta|x )\frac{x^{ 2b}}{(x^2+2)^{a_{\pi}+b}}  \, dx . 
	\]
	By the symmetry of $g(\beta)$ and Fubini's Theorem, we have
	\begin{eqnarray*}
		&& 1 - \int_{-k_n}^{k_n} g(\beta)\, d\beta = 2\int_{k_n}^{\infty} g(\beta)\, d\beta  \\
		&& =   2 \frac{2^{a_{\pi}}\Gamma(a_{\pi}+b)}{\Gamma(a_{\pi})\Gamma(b)}    \int_{ 0}^\infty    \int_{k_n}^{\infty}   \exp (- |\beta|x )  \, d\beta   \frac{x^{ 2b}}{(x^2+2)^{a_{\pi}+b}}  \, dx  \\
		&& =   2 \frac{2^{a_{\pi}}\Gamma(a_{\pi}+b)}{\Gamma(a_{\pi})\Gamma(b)}    \int_{ 0}^\infty    \frac{ \exp (- k_n x )}{x}   \frac{x^{ 2b}}{(x^2+2)^{a_{\pi}+b}}  \, dx  \\
		&& =      \frac{2 \Gamma(a_{\pi}+b)}{\Gamma(a_{\pi})\Gamma(b)}      \int_{ 0}^\infty     { \exp (- k_n \sqrt{2}x )}     \frac{x^{2b-1}}{(x^2 +1)^{a_{\pi}+b}}  \, dx  \\
		&& =   \frac{2 \Gamma(a_{\pi}+b)}{\Gamma(a_{\pi})\Gamma(b)}    \frac{1}{2\sqrt{\pi} \Gamma(a_{\pi} +b)} G_{1,3}^{3,1}\left(\frac{k_n^2}{2}
		\left\vert ^{1-b}_{a_{\pi}, 0, \frac{1}{2}} \right. \right)   \\
		&& =   \frac{1}{\sqrt{\pi}\Gamma(a_{\pi})\Gamma(b)}   G_{1,3}^{3,1}\left(\frac{k_n^2}{2}
		\left\vert ^{1-b}_{a_{\pi}, 0, \frac{1}{2}} \right. \right) 
	\end{eqnarray*}
	where the    second to last  ``='' follows from equation 3.389.2 in \cite{zwillinger2014table}.  
	The right side of the above equation looks similar to the marginal density of R2-D2 prior,  so  we can apply exactly the same technique used  in proof of Theorem  \ref{theorem_center properties R2-D2}. So   in the proof 
	\[1 - \int_{-k_n}^{k_n} g(\beta)\, d\beta  = 
	\frac{1}{\sqrt{\pi}\Gamma(a_{\pi})\Gamma(b)}   G_{1,3}^{3,1}\left(\frac{k_n^2}{2}
	\left\vert ^{1-b}_{a_{\pi}, 0, \frac{1}{2}} \right. \right)   
	= C_1^\ast U_1(k_n^2) + C_2^\ast U_2(k_n^2) + C_3^\ast U_3(k_n^2), 
	\]
	where 
	\begin{eqnarray*}
		C_1 ^\ast  &=& \frac{1}{\sqrt{\pi}\Gamma(a_{\pi})\Gamma(b)}   \Gamma(-a_{\pi})\Gamma(\frac{1}{2}-a_{\pi})\Gamma(a_{\pi} + \frac{1}{2}) \Gamma(1+a_{\pi})  <0 , \\
		C_2^\ast  &=& \frac{1}{\sqrt{\pi}\Gamma(a_{\pi})\Gamma(b)}  \Gamma(a_{\pi} ) \Gamma(\frac{1}{2})\Gamma(\frac{1}{2}) \Gamma(1-a_{\pi}) >0  , \\
		C_3^\ast   &=& \frac{1}{\sqrt{\pi}\Gamma(a_{\pi})\Gamma(b)}    \Gamma(a_{\pi}-\frac{1}{2})\Gamma(-\frac{1}{2})\Gamma(\frac{3}{2})\Gamma(\frac{3}{2}-a_{\pi}) > 0 , \\
		%		\end{eqnarray*}
		%		\begin{eqnarray*}
		U_1(k_n^2) &=& 
		%		 		\sum\limits_{j=0}^{\infty} \frac{\Gamma(a_{\pi n}+b+j) }{\Gamma(1+a_{\pi n}+j)\Gamma(\frac{1}{2} + a_{\pi n}+j) } \frac{(-1)^j(\frac{k_n^2}{2})^{j+a_{\pi n}}} {j!}  \equiv	
		\sum\limits_{j=0}^{\infty}(-1)^j  u_1(j, k_n^2),  \ 
		u_1(j, k_n^2) =  \frac{\Gamma(a_{\pi}+b+j) }{\Gamma(1+a_{\pi}+j)\Gamma(\frac{1}{2} + a_{\pi}+j) } \frac{ (\frac{k_n^2}{2})^{j+a_{\pi}}} {j!} , \\
		U_2(k_n^2) &=& 	 
		%		 	\sum\limits_{j=0}^{\infty} \frac{\Gamma(b+j) }{\Gamma(1-a_{\pi n}+j)\Gamma(\frac{1}{2}+j)} \frac{(-1)^j(\frac{k_n^2}{2})^j}{j!}  \equiv
		\sum\limits_{j=0}^{\infty} (-1)^j u_2(j, k_n^2),  \
		u_2(j, k_n^2)	 =\frac{\Gamma(b+j) }{\Gamma(1-a_{\pi}+j)\Gamma(\frac{1}{2}+j)} \frac{ (\frac{k_n^2}{2})^j}{j!} , \\
		\text{ and }
		U_3(k_n^2) &=& 	
		%		 		\sum\limits_{j=0}^{\infty}  \frac{\Gamma(\frac{1}{2}+b+j)}{\Gamma(\frac{3}{2}-a_{\pi n}+j)\Gamma(\frac{3}{2}+j)}
		%		 	\frac{(-1)^j(\frac{k_n^2}{2})^{j+1/2}}{j!} \equiv	
		\sum\limits_{j=0}^{\infty} (-1)^ju_3(j, k_n^2),  \ 
		u_3(j, k_n^2) =  \frac{\Gamma(\frac{1}{2}+b+j)}{\Gamma(\frac{3}{2}-a_{\pi}+j)\Gamma(\frac{3}{2}+j)}
		\frac{(\frac{k_n^2}{2})^{j+1/2}}{j!}. 
	\end{eqnarray*}
	Then we have  \[
	u_1(0,k_n^2) -  u_1(1,k_n^2)   \leq U_1(k_n^2)   \leq  u_1(0,k_n^2)  ,
	\]
	\[
	U_2(k_n^2) \leq  u_2(0,k_n^2) ,
	\]
	\[
	U_3(k_n^2)  \leq  u_3(0,k_n^2)  . 
	\]
	Hence,   based on the fact  $\Gamma(1-z) = -z\Gamma(-z)$,  $\Gamma(z) \approx 1/z$ as $z\rightarrow0$,   $k_n \asymp \sqrt{(q_n\log p_n)/n}/p_n\rightarrow 0 $,  and $a_{\pi}  \leq  \frac{\log(1-p_n^{-(1+u)})}{2\log k_n}\rightarrow 0$, it follows 
	\begin{eqnarray*}
		1 - \int_{k_n}^{k_n} g(\beta)\, d\beta  & \leq  &
		C_1^\ast \Big\{  u_1(0,k_n^2) -  u_1(1,k_n^2)  \Big \} + C_2^\ast u_2(0,k_n^2)   + C_3^\ast u_3(0,k_n^2)  \\
		&=&  1   - k_n^{2a_{\pi}} \Big\{  -\frac{\Gamma(-a_{\pi})\Gamma(\frac{1}{2} - a_{\pi})\Gamma(a_{\pi}+b)}{\sqrt{\pi}\Gamma(a_{\pi})\Gamma(b) 2^{a_{\pi}}}
		- C_4^\ast k_n^2 - C_5^\ast  k_n^{1-2a_{\pi}}  		   \Big\}\\
		&=&  1   - k_n^{2a_{\pi}} \Big\{  \frac{\Gamma(1-a_{\pi})\Gamma(\frac{1}{2} - a_{\pi})\Gamma(a_{\pi}+b)}{\sqrt{\pi} \Gamma(b) 2^{a_{\pi}}}
		- C_4^\ast k_n^2 - C_5^\ast k_n^{1-2a_{\pi}}  		   \Big\}\\
		& \rightarrow &  1   -    k_n^{2a_{\pi}} \leq p_n^{-(1+u)}, 
	\end{eqnarray*}
	where  $C_4^\ast, C_5^\ast\geq 0$.  Hence for  $a_{\pi}  \leq  \frac{\log(1-p_n^{-(1+u)})}{2\log k_n}\rightarrow 0$, we proved the first condition in   (\ref{condition_b1}) holds.

	%			where the last $\approx$ is obtained from Lemma B.1. from \cite{bai2018beta}, i.e., if $a_{\pi n}>1/2$ and  $b\rightarrow 0 $ as $n\rightarrow \infty$, then $\Gamma(a_{\pi n}+b)/( {\Gamma(a_{\pi n})\Gamma(b)} ) \approx b $. 
	
	Now let's prove that the second condition  in  (\ref{condition_b1}) also holds.  
	%By the symmetry of $g(\beta)$, the infimum of $g(\cdot)$ on the interval $[-E_n, E_n]$ occurs at either $-E_n$ or $E_n$. Then 
	%$
	%\inf _{ \beta \in [-E_n, E_n]} g(\beta)  =  g(E_n). 
	%$
	By Theorem \ref{theorem_tail properties}, $g(\beta)  = O(|\beta|^{-2b-1})$ as $|\beta|\rightarrow\infty$.  Since $\log(E_n ) = O(\log p_n)$, 
	\[
	\inf _{ \beta \in [-E_n, E_n]} g(\beta) = g(E_n) \approx O(E_n^{-2b-1}). 
	\]
	%\begin{eqnarray*}
	%&&\inf _{ \beta \in [-E_n, E_n]} g(\beta)   \\ 
	% && \geq \frac{2^{a_{\pi n}}\Gamma(a_{\pi n}+b)}{\Gamma(a_{\pi n})\Gamma(b)}   \int_{ 0}^\infty   \exp (- E_n x )\frac{x^{ 2b}}{(x^2+2)^{a_{\pi n}+b}}  \, dx   \\
	% &&  \approx  2^{a_{\pi n}} b   \int_{ 0}^\infty   \exp (- E_n x )\frac{x^{ 2b}}{(x^2+2)^{a_{\pi n}+b}}  \, dx   \\
	%  &&   = 2^{a_{\pi n}} b   \int_{ 0}^\infty   \exp (- E_n x )( \frac{x^{ 2}} {x^2 + 2})^{ b} (\frac{1}{x^2+2})^{a_{\pi n}}  \, dx   \\
	%  &&  \geq 2^{a_{\pi n}} b   \int_{ 1/2}^1 \exp (- E_n x )( \frac{x^{ 2}} {x^2 + 2})^{ b} (\frac{1}{x^2+2})^{a_{\pi n}}  \, dx   \\
	%    &&  \geq  ... \geq 1/p_n  ??????? \\
	%\end{eqnarray*}
	So 
	\[
	-\log\big(\inf _{ \beta \in [-E_n, E_n]} g(\beta)  \big)  =  (2b+1)\log(E_n) =   O( \log p_n), 
	\]
	i.e.,  the second condition in  (\ref{condition_b1}) holds. Hence, all conditions in  Lemma  \ref{lemmaB1} are satisfied. 
	Theorem  \ref{theorem_new theorem}  is proven. 
\end{proof}

\section{Appendix: Choosing Hyper-parameters $\nu$ and $\mu$}

Assume that $q$ of the $p$ variance components $(\lambda_j)$ account for $(1-\epsilon)$ proportion of the
variability. That is, $\frac{\sum_{j=1}^q \lambda_{(j)}}{\sum_{j=1}^p \lambda_{(j)}} = 1-\epsilon$, where $\lambda_{(j)}$ is the j-th largest value. Suppose that $\lambda_j \sim$ Gamma $(\nu, \mu)$. The degree of sparsity, i.e. the number of relevant components, $q$ can then be viewed as finding $q$ such that the proportion achieves $1-\epsilon$. Since the proportion is a random quantity, we seek $q$ such that the median of the distribution of the proportion reaches equality. Note that since, $\mu$ is a scale parameter, the distribution of the proportion is unaffected by $\mu$. Hence a grid search can be achieved to find $\nu$ such that the equation is satisfied by the median of the distribution, this just requires sampling from a Gamma $(\nu, 1)$ distribution to approximate the distribution for a given $\nu$. 

After choosing $\nu$, we then choose the scale parameter, $\mu$, we seek to ensure that the prior distribution for $\betabf$ is near zero at $p - q$ directions. To do so, recall that the distribution must lie on the ellipsoid given by $\betabf^T \Sigma_X \betabf = \theta \sigma^2$. Hence, we wish to push $p-q$ directions to be inside the ellipsoid, i.e. near zero. So, we set $\mu$ such that the probability of $\beta_j$ under the unrestricted Normal distribution, being inside the ellipsoid is $1 - q/p$. Note that this probability is just $P(\beta_j^2 < \theta \sigma^2)$. Now 
\begin{align*}
P(\beta_j^2 < \theta \sigma^2) =& \ P\left ( w_j < \frac{1}{\lambda_j} \right ) = \int_{0}^{\infty} \int_{0}^{1/n\lambda_j} f_w (w_j) f_\lambda(\lambda_j) dw_j d\lambda_j \\
=& \ \int_{0}^{\infty} \frac{1}{\Gamma(1/2)} \gamma \left ( \frac{1}{2},\frac{1}{2\lambda_j} \right ) f_\lambda(\lambda_j) d\lambda_j \\
=& \ \frac{1}{\Gamma(1/2)} E_\lambda \left [ \gamma \left ( \frac{1}{2},\frac{1}{2\lambda_j} \right ) \right ],
\end{align*}
where $w_j$ is a $\chi^2_1$ random variable, and $\gamma$ is the lower incomplete gamma function. This can again be evaluated on a grid. This time a grid of $\mu$.
	
\end{document}